\documentclass[11pt]{article}

\usepackage{ifxetex,ifluatex}
\if\ifxetex T\else\ifluatex T\else F\fi\fi T%
  \usepackage{fontspec}
\else
  \usepackage[T1]{fontenc}
  \usepackage[utf8]{inputenc}
  \usepackage{lmodern}
\fi

\usepackage{hyperref}

\usepackage{ifxetex,ifluatex}
\newif\ifxetexorluatex
\ifxetex
  \xetexorluatextrue
\else
  \ifluatex
    \xetexorluatextrue
  \else
    \xetexorluatexfalse
  \fi
\fi

\ifxetexorluatex
  \usepackage{fontspec}
\else
  \usepackage[T1]{fontenc}
  \usepackage[utf8]{inputenc}
  \usepackage{lmodern}
\fi
\usepackage{amsthm}
\usepackage{paithan}
\usepackage{xcolor}
\usepackage{mathtools}
\usepackage{tikz}
\usepackage{enumitem}
\usepackage[ruled,lined]{algorithm2e}
\usetikzlibrary{arrows.meta}
\usetikzlibrary{shapes, fit, backgrounds, positioning}

\usepackage{braket}

\theoremstyle{plain}
\newtheorem{definition}{Definition}[section]
\newtheorem{theorem}{Theorem}[section]

\newtheorem{corollary}{Corollary}[section]
\newtheorem{lemma}{Lemma}[section]

\newtheorem{observation}{Observation}[section]
\newtheorem{openquestion}{Open Question}

\usepackage{geometry}
\usepackage{setspace}
\newcommand{\outP}{\outcomeClass{P}}
\newcommand{\outN}{\outcomeClass{N}}

\newcommand{\rsGeog}{\ruleset{Generalized Geography}}
\newcommand{\rsUndirGeog}{\ruleset{Undirected Geography}}

\def\nimber{\mbox{\rm nimber}}
\def\mex{\mbox{\rm mex}}

\clearpage
\title{Winning the War by (Strategically) Losing Battles: Settling the Complexity of Grundy-Values in {\sc Undirected Geography}}

\author{Kyle W. Burke
\\Plymouth State\\kwburke@plymouth.edu \and
Matthew T. Ferland
\\ USC\\mferland@usc.edu \and Shang-Hua Teng\thanks{%Department of Computer Science and Department of Mathematics, Universiry of Southern California (USC).
Supported by the Simons Investigator Award for fundamental \& curiosity-driven research and NSF grant CCF-1815254.}\\ USC\\
shanghua@usc.edu}

\begin{document}
\maketitle

% Enable SageTeX to run SageMath code right inside this LaTeX file.
% http://doc.sagemath.org/html/en/tutorial/sagetex.html
% \usepackage{sagetex}

% Enable PythonTeX to run Python – https://ctan.org/pkg/pythontex
% \usepackage{pythontex}

\begin{abstract}
We settle two long-standing complexity-theoretical questions---open since 1981 and 1993---in combinatorial game theory (CGT).

We prove that the {\em Grundy value} (a.k.a. {\em nim-value}, or {\em nimber}) of \ruleset{Undirected Geography} is \cclass{PSPACE}-complete to compute.
This exhibits a stark contrast with a result from 1993 that \ruleset{Undirected Geography} is {\em polynomial-time solvable}.
By distilling to a simple reduction, our proof further establishes a {\em dichotomy theorem}, providing
a ``phase transition to intractability'' in Grundy-value computation, sharply characterized by a maximum degree of four:
The Grundy value of \ruleset{Undirected Geography} over any degree-three graph is polynomial-time computable, but over degree-four graphs---even when planar \& bipartite---is \cclass{PSPACE}-hard.
Additionally, we show, for the first time, how to construct \ruleset{Undirected Geography} instances with Grundy value $*n$ and size polynomial in $n$.

We strengthen a result from 1981 showing that sums of tractable partisan games are \cclass{PSPACE}-complete in two fundamental ways.  First, since \ruleset{Undirected Geography} is an impartial ruleset, we extend the hardness of sums to impartial games, a strict subset of partisan.  Second, the 1981 construction is not built from a natural ruleset, instead using a long sum of tailored short-depth game positions.  We use the sum of two \rsUndirGeog\ positions to create our hard instances.
Our result also has computational implications to Sprague-Grundy Theory (1930s) which shows that the Grundy value of the {\em disjunctive sum} of any two impartial games can be computed---in polynomial time---from their Grundy values.
In contrast, we prove that assuming $\cclass{PSPACE}\neq \cclass{P}$, there is no general polynomial-time method to summarize two polynomial-time solvable impartial games to efficiently solve their disjunctive sum.
\end{abstract}

\newpage

\section{Introduction}

Knowing how to win battles does not always translate into knowing how to win wars. More often than not, the victor must strategically lose some winnable battles in order to win the war.
This timeless principle is elegantly captured in the celebrated Sprague-Grundy Theory about impartial games from the 1930's
in combinatorial game theory (CGT)  \cite{Sprague:1936,Grundy:1939,WinningWays:2001}.
This theory introduces the concept of the {\em Grundy value}, and uses it---together with Bouton's constructive theory \cite{Bouton:1901} on \ruleset{Nim}\footnote{A \ruleset{Nim} game  starts with a collection of piles of items.
Two (or multiple players)
takes turns picking at least one items from one of the piles.
Under normal play, the player taking the last items wins the game.
\ruleset{Nim} was known in ancient China as Jian Shi Zi (picking pebbles).
}---to
characterize the winning strategy for the {\em disjunctive sum} of multiple ``battlefield'' games.
In this paper, we resolve a long-standing complexity-theoretical question---open since 1981---concerning the {\em computational complexity of strategic losing} for the goal of winning the overall sum game (the disjunctive sum).
As the main technical result of this paper, we settle another question in CGT---open since 1993---on the {\em complexity barrier} of Grundy values of a well-studied impartial graph-theoretical ruleset.
Our theoretical work has also inspired new practical board games.

\subsection{Games of Games: Disjunctive Sum}

A combinatorial game is defined by a succinct {\em ruleset}, specifying the domain of {\em game positions}, and
for each position, the set of {\em feasible options} each player can move the game to \cite{WinningWays:2001,LessonsInPlay:2007,SiegelCGT:2013,DBLP:books/daglib/0023750}.
%\footnote{A game position is usually referred to as a game, a position, or an instance of a ruleset.}
A ruleset is {\em impartial} if both players have the same options at every position.
Games that aren’t impartial are known as {\em partisan}.
In the {\em normal-play} setting, two players
take turns advancing the game, and the player
  who is forced to start their turn on a {\em terminal position}---a position with
   no feasible options---loses the game.
We combine the challenges of deciding the \emph{winnability} and selecting a \emph{winning move} (whenever one is available) into one term: \emph{strategic determination}. (See also \cite{fraenkel2004complexity} for integrating  the two tasks).
If the player with a winning strategy can consistently solve the strategic determination problem, then they can play the game optimally and win.

For computational analysis, a {\em size} is associated with each
game position---e.g. % the number of vertices in the graph for \ruleset{Node Kayles}\footnote{\ruleset{Node Kayles} is played on a graph; feasible options are graphs obtained by removing a node and its neighbors.} or
of bits encoding \ruleset{Nim}---as the basis for measuring complexity  \cite{DBLP:journals/jcss/Schaefer78,fraenkel2004complexity,PapadimitriouBook:1994,BurkeFerlandTengQCGT}.
The size measure is assumed to be {\em natural} \footnote{In other words, the naturalness assumption
rules out rulesets with embedded hard-to-compute predicate like---as a slightly dramatized illustration---``If Riemann hypothesis is true, then the feasible options of a position include removing an item from the last \ruleset{Nim} pile.''}. with respect to the key components of the ruleset. In  particular, at each position with size $n$:  (1) the space of feasible options can be identified in time polynomial in $n$, and (2) all positions reachable from the position have size upper-bounded by a polynomial function in $n$.
An impartial ruleset is said to be {\em polynomial-time solvable}---or simply, {\em tractable}---if there is a polynomial-time algorithm for its strategic determination.
Multiple games %, even from different rulesets,
can be combined into a new game:
%when played together in parallel:

\begin{definition}[Disjunctive Sum]
For any two combinatorial games  $G$ and $H$ (respectively, of rulesets $R_1$ and $R_2$), their {\em disjunctive sum},  $G+H$, is
a game in which %at each turn,
the next player chooses to make a move in exactly one of $G$ and $H$, leaving the other alone.
A sum game $G+H$ is terminal if and only if
   both $G$ and $H$ are terminal according to their %corresponding
   own rulesets.
   Recursively then, in a sum of three or more, the current player still chooses to move in exactly one of the components.
\end{definition}

\subsection{Outstanding Computational Questions About Sprague-Grundy Theory}
\label{Sec:Open}

In 1981, Morris \cite{morris1981playing} proved that the sum of tractable {\bf\em partisan} games can be \cclass{PSPACE}-hard.
His theorem elegantly encapsulates the fundamental intricacy of strategic interaction among (simple) games (even for introductory textbooks to the field \cite{LessonsInPlay:2007}).
Morris generates lists of individually-tractable partisan games that combine to create an intractable sum.
It is used as a starting point for other stronger versions of the proof, such as limiting the games to be depth 2 \cite{yedwab1985playing} and limiting the number of branches for each game to 3 \cite{moews1993some}.
%These improvements
A further adaptation is made to show that \ruleset{Go Endgames} are \cclass{PSPACE}-hard \cite{wolfe2000go}.
%'s result provides technical justification for why the concept known as ``switches'' makes partisan games difficult, as well as a strong basis to reduce from for other proofs, such as proving \ruleset{Go Endgames} to be PSPACE-hard \cite{wolfe2000go}.
Morris' theorem also provides a framework for understanding other families of combinatorial games. One important basic
 question has remained open since Morris' construction:

\begin{openquestion}[Sum of Impartial Games]\label{Dsum}
Can the disjunctive sum of two (or more) polynomial-time solvable {\bf\em impartial games} become intractable?
\end{openquestion}

This open question is fundamentally connected with Sprague-Grundy theory, a seminal part of CGT. Formulated in 1930s,  this theory provides a mathematical characterization for impartial games and their sum \cite{Sprague:1936,Grundy:1939}, laying the foundation for modern CGT \cite{ONAG:2001,WinningWays:2001}.
We now highlight two fundamental aspects of this beautiful theory
(and provide the background of our own work):

\vspace{0.1in}
\noindent
%\begin{itemize}
%\item
{\bf\em Concise Mathematical Summary of Impartial Games}: Playing combinatorial games optimally usually requires deep strategic reasoning about long alternation down the last level of their {\em game trees}.
Remarkably, Sprague and Grundy showed that the essence of every impartial game can be distilled into an ``equivalent''
single-pile \ruleset{Nim} game.
%Each position's
Its {\em Grundy value} ({\em a.k.a.} {\em nim-value} or {\em nimber}) is then the number of items in the equivalent single-pile \ruleset{Nim}.
The nim-value extends
%contains sufficient information for determining the %game's
{\em winnability}: the current player has a winning strategy
if and only if the %game's
Grundy value is not equal zero.
To win, it is sufficient to choose any feasible option with  value zero.
The Grundy value of a game provides a {\em succinct mathematical summary} of its game tree, whose size could be
{\em exponential} in the number of options:
the Grundy value is always {\em bounded above} by the number of options.

%\item
\vspace{0.1in}
\noindent {\bf\em Systematic Framework for Combining Games}:
%Sprague and Grundy discovered that, beyond winnability, the nim-values capture strategic information in {\em disjunctive sums} of impartial games.
%played in parallel, in which on each turn the current player can move in exactly one of the two games.
Sprague and Grundy's pioneering theory establishes a systematic framework not only for combining games across different rulesets, but also for a complete characterization of strategic interaction among games in the overall sum, based only on their concise summaries %(i.e., the Grundy-values)
\cite{WinningWays:2001,ONAG:2001}.
Combined with Bouton's theory on \ruleset{Nim} \cite{Bouton:1901} (see Appendix \ref{AppendixNim} for a description of \ruleset{Nim} and how to solve it in polynomial time),
letting $\oplus$ denote the {\bf\em bitwise xor} (the {\bf\em nim-sum}), the theory establishes:
%defines an algebraic system %over natural numbers that %completely
%characterizing strategical interactions in the disjunctive sum:
\begin{eqnarray}
\nimber(G+H) = \nimber(G)\oplus \nimber(H) \quad\quad
\mbox{$\forall$ impartial $G$, $H$}
\end{eqnarray}

In general, losing a winnable game may be
 necessary to win in the overall sum (true to the meaning of ``losing a battle but winning the war'').
%The nimber-based characterization also  shows that the search for a winning strategy in the sum game can be mathematically guided by the nim-values of the evolving battlefield games.
Sprague Grundy theory  contains a profound computational statement, made prior to the inception of \cclass{P} {\em vs} \cclass{NP}.
Because the nim-sum  is  linear-time computable,
%the following reduction is polynomial-time:
if the Grundy values of the games are %polynomial-time
tractable, then the Grundy value---hence the winnability---of the overall sum game is also tractable.
%Note that when one has the Grundy values for all individual games, one may always find the Grudy value, and thus winnability, of the sum of those games in polynomial time.
This contrasts with some values in partisan games, as exploited by Morris \cite{morris1981playing}, where he constructs a CGT representation for each component, but proves one is unable to %compute the sum in
``add them up'' in polynomial time (unless P = PSPACE).

The following open question is intrinsically connected with Open Question \ref{Dsum}.

\begin{openquestion}
Can the information captured in the Grundy value of an impartial game be more  expensive to compute than its strategic determination?
\end{openquestion}
Our research has been influenced by the following two
tightly related yet subtly different
%of complexity-theoretical
formulations  concerning the  algorithmic connection between %these two basic solution concepts.
Grundy values and strategic determination.

\begin{openquestion}[Tractable Structures]\label{Structures}
For any ruleset, does polynomial-time strategic determination  imply polynomial-time Grundy-value computation?
\end{openquestion}

\begin{openquestion}[Efficient Reduction]\label{Reduction}
Is there a general polynomial-time reduction from  Grundy-value computation to strategic determination?
\end{openquestion}

These last two questions are directly related - a YES answer to the second affirms the first (and thus a NO answer to the first also refutes the second.)
An efficient solution for these two would also provide a unified algorithm---based on Sprague-Grundy theory---for efficiently solving the disjunctive sum of tractable impartial games (and hence Open Question \ref{Dsum}).

On the tractable spectrum, some polynomial-time solvable rulesets---including \ruleset{Nim}, \ruleset{Subtraction Game} \cite{WinningWays:2001}, and many others---have {\em dual tractability}: their Grundy values are also polynomial-time computable.
Open Question \ref{Structures} focuses on whether
strategic determination and Grundy values
have common underlying mathematical structures for  tractability beyond the fact
that both %While both strategic determination and the Grundy value
can be obtained by evaluating the game tree \cite{WinningWays:2001}.
This is relevant to the part of Fraenkel's work \cite{fraenkel2004complexity}, where he conceptualized
 a class called {\em games with an efficient strategy}
 by combining the tractability of their own strategic determination with the tractability of their involvements in disjunctive sums and in mis\`{e}re-play
 (that is,  the current player wins at any terminal position).

%\footnote{Here, we would like to emphasize that both questions are about a ruleset, that is, the whole family of game instances of the ruleset rather than individual game instances.}
On the intractable spectrum, for any \cclass{PSPACE}-complete game with polynomial game-tree height---e.g.,  \ruleset{Node-Kayles} \cite{DBLP:journals/jcss/Schaefer78}, \rsGeog\ \cite{DBLP:journals/jcss/Schaefer78,LichtensteinSipser:1980}, \ruleset{Col} \cite{ONAG:2001,DBLP:journals/tcs/BeaulieuBD13} and many games based on logic, topology, network sciences,
{\em etc}, \cite{DBLP:journals/jcss/Schaefer78,DBLP:journals/im/BurkeT08,PosetGame,Grier,burke2021transverse}---the answer to the second question is always a YES.
%\footnote{The affirmative answer comes from the fact that the Grundy value of games in these rulesets can be computed in space polynomial in the height of their game trees, by the standard DFS-traversal of the game trees. \cite{PapadimitriouBook:1994}}
%A YES answer to the second question for tractable impartial games would have given a unified algorithmic solution to the disjunctive sum of tractable games.
However, this complexity-theoretical polynomial-time reduction is not extendable from \cclass{PSPACE}-complete games to games with potentially lower complexity.
In addition to tractable games,  it remains open whether a polynomial-time reduction from Grundy values  to strategic determination exists for intractable impartial games, whose complexity might be ``strictly'' in-between \cclass{NP} and \cclass{PSPACE}.\footnote{Is there a polynomial-time nimber-to-winnability reduction for impartial games---arising in quantum combinatory game theory \cite{BurkeFerlandTengQCGT}---for which
%deciding whether or not the current player has a winning strategy
strategic determination is complete for a particular level of the polynomial-time hierarchy?}
Open Question \ref{Reduction} hypothesizes whether a unified algorithmic approach exists for Grundy-value computation using winnability testing \& winning-move finding
%winning decision---i.e., determining winnability and a winning move when available---
as subroutines.

As Fraenkel pointed out, winnability alone may not capture the whole picture of game's tractability \cite{fraenkel2004complexity}.
Recent progress on %the complexity of
\cclass{Poset} games \cite{PosetGame}\footnote{A \ruleset{Poset} game is a two-player impartial game over a {\em partially ordered set} (poset), in which
each move---the selection of an element in the poset---removes it together with all elements that are greater. The \ruleset{Poset} game generalizes
the classical ``chocolate-eating'' game \ruleset{Chomp} \cite{Zeilberger:2004} as well as %``stone picking''
\ruleset{Nim}.
A poset with the greatest element is a poset that contains an element greater than any other element in the poset.}
highlights that aspect as well.
It is well-known---by {\em strategy-stealing}---that the first player has a winning strategy in any \ruleset{Poset} game where the underlying poset has a {\em greatest element} (e.g. in \ruleset{Chomp}), providing
straightforward answer to winnability.
On the other hand, Bodwin and Grossman \cite{bodwin2019strategystealing}
prove that in this family, finding a winning move %in such \ruleset{Poset} games
can be \cclass{PSPACE}-complete.
Hence even in this special case, the Grundy-value and strategic determination are polynomial-time reducible to each other. %, because winning decision needs to return a winning move.
The implication on the nimber-winnability complexity separation also has a caveat.
\ruleset{Poset} games with the greatest element may have reachable game positions without the greatest element:
playing these special %\ruleset{Poset}
games %may
requires later moves on
``normal'' \ruleset{Poset} games, outside the greatest-element family.
Indeed, Grier \cite{Grier} proves that deciding winnability of normal \ruleset{Poset} games is \cclass{PSPACE}-complete.

\subsection{Battles of Geography Without Directions: A Concrete Open Question}

As a game version of the {\em Seven Bridges of
K\"{o}nigsberg}, \ruleset{Geography} grew from a real-world ``Word Chain'' game---with cities as the category---into   an abstract game on graphs, as suggested by Richard Karp \cite{DBLP:journals/jcss/Schaefer78}.
This  game, known as \ruleset{Generalized Geography},  became the main subject for complexity study in the landmark paper, ``GO is polynomial-space hard'' (1978) by Lichtenstein and Sipser \cite{LichtensteinSipser:1980}.
In this impartial game, a position is defined by
a directed graph and a specified node (with the token). During the game, two players take turns moving the token to an outgoing neighbor and removing the node it just occupied (or otherwise marking it so the token cannot re-visit any node).
In the normal-play setting, the player who cannot make a move on their turn loses the game.
\rsGeog\ was originally shown to be \cclass{PSPACE}-complete by Schaefer \cite{DBLP:journals/jcss/Schaefer78}; this was improved by Lichtenstein and Sipser \cite{LichtensteinSipser:1980} to be \cclass{PSPACE}-complete, even when the graph is planar, bipartite, and has a maximum degree of three.
%Through a sequence of carefully crafted gadgets, Schaefer \cite{DBLP:journals/jcss/Schaefer78}, Lichtenstein and Sipser \cite{LichtensteinSipser:1980} proved that deciding the winnability in \ruleset{Generalized Geography}  is \cclass{PSPACE}-complete, even when the graph is planar, bipartite, and has a maximum degree of three.
These graph properties are essential to their analysis of \ruleset{Go}, whose game board is a 2D grid.

In 1993,  Fraenkel, Scheinerman, and Ullman \cite{DBLP:journals/tcs/FraenkelSU93} added a new twist. They proved that \ruleset{Undirected (Vertex) Geography}---the special case %of \ruleset{Generalized Geography}
over undirected graphs---is polynomial-time solvable.
In 2015, Renault and Schmidt \cite{MisereUG}
 revitalized interest in \rsUndirGeog\ by showing that it's \cclass{PSPACE}-complete under mis\`{e}re-play
 %---i.e., the current player wins at any terminal position---
 instead of normal play.
Both the edge variant \cite{DBLP:journals/tcs/FraenkelSU93} and short version \cite{ShortUG} of \ruleset{Undirected Geography} are also shown to \cclass{PSPACE}-complete.
Various extensions to \ruleset{Undirected Geography}
%, e.g., partisan formulation, three-players
has been analyzed \cite{venkataraman2001survey,bosboom2020edge,matsumoto2020feedback}.

The Fraenkel-Scheinerman-Ullman solution
%for determining winnability \& identifying winning moves
is guided by an elegant matching theory and supported by efficient matching algorithms (see Appendix \ref{Appendix:Matching}).
For any undirected graph $G=(V,E)$  and $s\in V$ satisfying
  $E\neq \emptyset$,
the current player has a winning strategy at \ruleset{Undirected Geography} position $(G,s)$ if and only if $s$ is in every maximum matching of $G$.
However, this matching-based characterization appears to be limited to winnability.
{\em Whether or not the Grundy value of \ruleset{Undirected Geography} is
polynomial-time computable
%tractable
had been elusive}.

\begin{openquestion}[Tractable or Intractable]
Is the Grundy value for \ruleset{Undirected Geography} computable in polynomial-time?
\end{openquestion}

\ruleset{Undirected Geography}  has thus become
an exemplary tractable impartial game for which
no efficient algorithm had been discovered for
its Grundy-value computation.
Others such as \ruleset{Moore's Nim} \cite{moore1910generalization,WinningWays:2001} and \ruleset{Wythoff's game} \cite{wythoff_1907} are also wonderful examples \cite{fraenkel2004complexity}.

\subsection{Our Contributions}

In this paper, we settle these open questions.

\vspace{0.1in}
\noindent {\bf A Dichotomy Theorem on Grundy Values}:
As our main technical result, we prove
that computing the Grundy value of
%polynomial-time solvable
\ruleset{Undirected Geography} is \cclass{PSPACE}-complete.
The key step is to {\em impose a direction} %of ``information flow''
over  undirected edges, where game %geographical
paths can travel across in either direction.
The complexity analysis has another intricacy -
%The polynomial-time winnability now becomes an obstacle---or metaphorically, a Bermuda Triangle in \ruleset{Undirected Geography}---for establishing \cclass{PSPACE}-hardness of  Grundy-value computation.
 the polynomial-time winnability is an obstacle.
%establishing \cclass{PSPACE}-hardness of Grundy~values
% needs to avoid zero-value positions.
% our constructions are forced to avoid zero-value positions
%like the Bermuda Triangle in the North Atlantic Ocean
%in order to establish \cclass{PSPACE}-hardess of Grundy~values. % computation.

By distilling our original complex construction into a simpler reduction using Lichtenstein-Sipser,  we are able to establish a {\em dichotomy theorem} \cite{SchaeferDichotomy}---in Section \ref{sec:Dichotomy}---providing
a ``phase transition to intractability'' in Grundy-value computation, sharply characterized by a maximum degree of four:

\begin{theorem}[Geographical Dichotomy]\label{thm:dichotomy}
The Grundy value of \mbox{\ruleset{Undirected Geography}} over degree-three graphs is polynomial-time computable
%can be computed in polynomial time,
but over degree-four graphs---even when planar \& bipartite---is \cclass{PSPACE}-hard.
\end{theorem}

In our polynomial-time Grundy-value algorithm for degree-three, the Fraenkel-Scheinerman-Ullman algorithm
\cite{DBLP:journals/tcs/FraenkelSU93} is applied to navigate a branch-and-bound process for evaluating game~trees.

\vspace{0.1in}
\noindent {\bf Strategic Losing is Hard}:
Our proof
also show that distinguishing  $\ast$ or $\ast 2$
in \ruleset{Undirected Geography}
is \cclass{PSPACE}-hard,
a detail crucial in our next theorem solving Open Question \ref{Dsum}:

\begin{theorem}[Strategic Synergy]\label{theo:SumIsHard}
The disjunctive sum of two tractable impartial games---in particular the sum of two \ruleset{Undirected Geography} games---can be
\cclass{PSPACE}-hard to solve.
\end{theorem}

Our result strengthens Morris' 1981 result \cite{morris1981playing} in two fundamental aspects.
First, we extend the \cclass{PSPACE}-hardness
from the sum of partisan games to the sum of impartial games, shedding new light on the computational facet of the Sprague-Grundy characterization (more below).
Second, Morris' construction is built from  a long sum of tailored short-depth game positions.
In contrast, we use %the sum of
two games of a natural, well-studied \rsUndirGeog\ impartial ruleset to create our hard sum.
Our construction is in fact {\rm robust}:
In our \cclass{PSPACE}-hard sum, one of the  games---provided with non-zero Grundy value---can even be arbitrarily chosen, say by an adversary.
%\begin{quote} \noindent {\sc More Than the Sum of its Parts} - {\em The disjunctive sum of two tractable impartial games can become intractable.} \end{quote}

Sprague-Grundy provides a barrier against closely mimicking Morris' construction in the realm of impartial games.  Since nim sums are efficiently computable, one cannot present a long list of shallow impartial games where the winnability of the sum is intractable.  We overcame this obstacle by instead summing two positions where the individual Grundy values (of at least one) are difficult to discern.  We are curious whether there is a well-known, tractable, strictly-partisan ruleset where determining the winnability of the sum of two positions is computationally hard.

\vspace{0.1in}
\noindent {\bf Mathematical-Computational Divergence in Sprague-Grundy Theory}:
Our complexity result on this concrete graph game has wider computational implications
in connection with Sprague-Grundy Theory.
The sharp contrast between the complexity of strategic determination and  Grundy values in \ruleset{Undirected Geography} illustrates fundamental
mathematical-computational divergence in Sprague-Grundy theory.
%of impartial games and their disjunctive sums.
When computational cost is no object,
%the
%Sprague-Grundy
%theory establishes that
the Grundy values are {\em effective and concise} mathematical summaries
of game trees for strategic reasoning in disjunctive sums. %(because the nim-sum can be computed in linear time).
However, as we have shown, this elegant mathematical summary could be \cclass{PSPACE}-hard to obtain, even for polynomial-time solvable games.
That is, assuming $\cclass{PSPACE}\neq \cclass{P}$,
  the Grundy values of combinatorial games capture provably richer and potentially hard-to-compute structures
   than just their solvability.
In fact, Theorem \ref{theo:SumIsHard} implies a broader impossibility statement:

\begin{theorem}[Summary]%[Strategic Summaries are Hard to Prepare]
Unless \cclass{P} = \cclass{PSPACE}, there is no general polynomial-time method to summarize two polynomial-time solvable impartial games to efficiently solve their disjunctive sum.
\end{theorem}

\vspace{0.1in}
\noindent {\bf Towards Practical Board Games}:
In Section \ref{sec:GoG}, we apply Sprague-Grundy theory to resolve the complexity of several rulesets  that generalize \ruleset{Undirected Geography}. We first show very basic extensions, including \ruleset{Multi-Token Undirected Geography} and \ruleset{Undirected Geography with Passes}. Then we demonstrate the versatility of the result, by showing that \ruleset{Uno Swap}, a minor modification of the tractable \ruleset{Uncooperative Uno}\cite{demaine2014uno}, is PSPACE-complete.
These results have potential applications to the practical design of board games based on \ruleset{Undirected Geography}, where the real world
 appreciates games with simple rules and positions, combined with
 deep strategic reasoning for winning moves \cite{WinningWays:2001,Eppstein,burke2021transverse}.
 Thus, the removal of edge directions from \ruleset{Generalized Geography}, while also retaining its \cclass{PSPACE}-hard complexity opens up several possibilities.
In Section \ref{Sec:FinalRemarks}, we discuss two further practical extensions using the standard \ruleset{GO} or \ruleset{Hex} game boards.
%of \ruleset{Undirected Geography}, that we are currently exploring and implementing.
For example, the web-version of one of our new games, \ruleset{Binary Undirected Geography}, can be played at \url{https://turing.plymouth.edu/~kgb1013/DB/combGames/twoBUG.html}.

\vspace{0.1in}
\noindent {\bf Graphs with Polynomial-High \ruleset{Undirected Geography} Nimbers}:
In Section \ref{sec:Math}, we give a constructive proof that, for any $n$, there exists a  polynomial-sized \ruleset{Undirected Geography} instance with Grundy value $n$. %, whose underlying graph has size polynomial in $n$.
Logarithmic Grundy values are realizable by trees with recursive structures, and  linear Grundy values can be achieved by directed graphs in \ruleset{Generalized Geography}.
%Our gadget design and reduction analysis requires careful treatments of augmenting paths in connection with the  region of \ruleset{Undirected Geography} with Grundy value zero (characterized by the matching theory of Fraenkel, Scheinerman, and Ullman \cite{DBLP:journals/tcs/FraenkelSU93}).
To the best of our knowledge, this is the first
polynomial Grundy value construction
for \ruleset{Undirected Geography}.
Using this construction, we prove in Theorem \ref{Thm:Mys} that any classifier for {\em positive} Grundy values in \ruleset{Undirected Geography} is \cclass{PSPACE}-hard.

\section{The Value of Games Beyond Winning}

\label{sec:Background}

In combinatorial game theory, a ruleset defines not just a single game, but many---possibly infinitely many---game instances (or \emph{positions}).
Playing games requires strategic reasoning of one's own options
as well as opponent's subsequent options, %and so on,
to answer the key problem on {\bf\em  winnability}:

\begin{definition}[Strategic Determination]
Given a game $G$ under a ruleset $R$,
  determine whether or not the current player in $G$ has a winning option, and if YES, return a winning option of $G$.
\end{definition}

This fundamental problem---commonly involving {\em deep} alternation---has been the subject of intense mathematical and computational studies \cite{WinningWays:2001,AlgGameTheory_GONC3,DBLP:books/daglib/0023750,PapadimitriouBook:1994,DBLP:journals/jcss/Schaefer78,EvenTarjanHex,DBLP:journals/tcs/DucheneR14,Reisch:1981,DBLP:journals/tcs/Fukuyama03,DBLP:journals/im/BurkeT08,LichtensteinSipser:1980,DBLP:journals/tcs/BeaulieuBD13}.
A ruleset $R$ defines a natural {\em game tree}, capturing this alternation for each of its
positions, $G$, by recursively branching with feasible options.
%, whose root is associated with position $g$ itself. The number of children that the root has is equal to the   number of feasible options at $g$. Each child is associated   with such a position directly {\em reachable} from $g$,   and its sub-game-tree defined recursively.
Thus, the game tree of $G$ contains all
  {\em reachable positions} of $G$ under ruleset $R$, with the leaves as the terminal positions.
%Each node in the game tree is defined by a sequence of ``feasible moves'' and is associated with a game position reachable from $g$.
% The leaves of the game tree
%are terminal positions.% under the normal-play setting.
%Mathematically, the outcome class of an impartial game position can be logically viewed as the outcome in a Boolean evaluation process of the game tree when we equal %$\outcomeClass{N}$ winning position as  \texttt{true} and
%``\texttt{true}''- the current player has a winning option''
%$\outcomeClass{P}$ losing position as  \texttt{false}; the internal nodes are evaluated by the NAND (NOT-AND) operator.

In the 1930s, Sprague and Grundy independently discovered a deep, yet basic,  mathematical structure underlying impartial games, referred to as ``Sprague-Grundy theory''.
%beyond their winnabilities.
This foundational discovery % for combinatorial games---now the widely celebrated Sprague-Grundy Theorem---
characterizes each impartial game $G$ %of an impartial ruleset
by a  natural number,  known as the {\em Grundy value}  (a.k.a {\em nim-value} or {\em nimber}) of the game.
Recursively, the Grundy value of $G$ is:
\begin{itemize}
\item  {\bf Terminal Position}: If $G$ is terminal, then $\nimber(G) = 0$.
\item {\bf Non-Terminal Position}: If $\{G_1,...,G_{\Delta}\}$ is the
  set feasible options of $G$, then:
  %$\nimber(g)$ is equal to the smallest integer missing from the set:
%$\left\{\nimber(g_1),...,\nimber(g_{\Delta})\right\}.$
%In other words,
\begin{eqnarray}
\nimber(G) = \mex\left(\left\{ \nimber(G_1),...,\nimber(G_{\Delta})
\right\}\right)
\end{eqnarray}
where $\mex$ is the {\em minimum excluded value},
  returning the smallest value of %the complement set:
  $$\mathbb{Z}^+\cup\{0\}\setminus \left\{
  \nimber(G_1),...,\nimber(G_{\Delta})
  \right\}.$$
\end{itemize}

%To highlight the numerical denotation of game values,
We will use the notation standard in combinatorial game theory for Grundy values: $*k$ for $k$, except that $*$ is shorthand for $*1$ and 0 is shorthand for $\ast 0$.\footnote{The reason for the $\ast 0 = 0$ convention is that it is equivalent to the integer zero in CGT.}

By grouping all %{\em winnable positions}---
positions with non-zero Grundy values into a class called ``Fuzzy'',
    impartial game positions can be partitioned into two \emph{outcome classes},
characterizing  winnability.
%\begin{itemize}
    %\item
    (1) \outcomeClass{N} (``Fuzzy'') - with positive Grundy values; the current (next) player always has a winning strategy.
%    \item
(2) \outcomeClass{P} (``Zero'') -  with zero Grundy value;
 the previous player always has a winning strategy.

\section{A Dichotomy Theorem}
%\label{sec:Intractability}

\label{sec:Dichotomy}

In this section, we prove Theorem \ref{thm:dichotomy}, setting up the Dichotomy Theorem of Grundy-value computation in \ruleset{Undirected Geography} based on
the local degree of intersection. Because  ``Zero'' ($\outP$) is polynomial-time distinguishable from ``Fuzzy'' ($\outN$) in \ruleset{Undirected Geography},
to establish the hardness, we need to show that the ``Fuzzy'' region is \cclass{PSPACE}-hard to classify.
By a (rather involved) reduction from {\em True Quantified Boolean Formula}, we proved that
$*$ and $*2$ are \cclass{PSPACE}-hard to distinguish.
While aiming for planar graphs, we distilled this construction, finding
 a simple gadget (Figure \ref{fig:directedEdgeGadget})
to obtain a simpler and direct reduction from \ruleset{Generalized Geography}.
For readers who may want to see
   more complex constructions for
``direction control'' in \ruleset{Undirected Geography}, we refer them to Section \ref{sec:Math} on nimber constructability.
There, replacing directed edges is more intricate because high Grundy values cannot be truncated as in the complexity analysis below, and
our attempts to simplify the proof haven't yet produced
the same outcome.

\subsection{PSPACE-Complete Grundy Values of \ruleset{Undirected Geography}}

\begin{theorem}[Complexity Separation of Winnability and Grundy Values]
\label{thm:nimberHardness}
The Grundy value of polynomial-time solvable \ruleset{Undirected Geography}
   is \cclass{PSPACE}-hard to compute.
\end{theorem}

%To prove the theorem, we will reduce from \ruleset{(Directed) Geography}.
Our reduction, $r$, takes a \ruleset{Generalized Geography} position $(G, s)$ and yields a \ruleset{Undirected Geography} position $r(G,s) = (G', s)$ where:

$(G', s)
\begin{cases}
    = \ast, &\text{ if } (G, s) = 0\ (\in \outP)\\
    \in \outN \setminus \{\ast\}, &\text{ if } (G,s) \in \outN
\end{cases}$

For readers unaccustomed to working with nimbers, we provide another characterization: the (disjunctive) sum of $(G', s)$ with a simple $\ast$ yields:

$(G', s) + \ast \in
\begin{cases}
    \outP, &\text{ if } (G,s) \in \outP\\
    \outN, &\text{ if } (G,s) \in \outN
\end{cases}$

Thus, the reader can consider the winnability of $(G, s)$ equivalent to $(G', s) + \ast$.  This is not difficult to conceptualize: in addition to the resulting \ruleset{Undirected Geography} position, we also include another game with exactly one move shared by both players.  Thus, exactly once in the sum, one of the two players can make the only move in that game instead of moving in \ruleset{Undirected Geography}.  These two characterizations are equivalent, as $\ast + \ast = 0$ and, whenever $j \geq 2$, $\ast + \ast j = \ast k$ where $k \geq 2$.  In Appendix \ref{sec:alternateProofs}, we provide proofs using this second winnability characterization.
The reduction itself contains two parts:
\begin{enumerate}
    \item %First,
    Modify $G$ so that each vertex $v \in V$ is given an adjacent singleton vertex, $v_0$, adjacent to no other vertices.  In other words, $\forall v \in V$ we will add vertex $v_0$ and (undirected) edge $(v, v_0)$.
    \item %Second,
 For each of the directed edges in $G$, $(x,y)$, we replace it with the gadget shown in Figure \ref{fig:directedEdgeGadget}.
\end{enumerate}

 \begin{figure}[h!]\begin{center}
    \begin{tikzpicture}[node distance=1cm]
        \node[draw, circle] (x) {$x$};
        \node[draw, circle] (a) [right=of x] {$a$};
        \node[draw, circle] (a0) [below of=a] {$a_0$};
        \node[draw, circle] (b) [right=of a] {$b$};
        \node[draw, circle] (c) [below right=of b] {$c$};
        \node[draw, circle] (c0) [below of=c] {$c_0$};
        \node[draw, circle] (f) [above right=of b] {$f$};
        \node[draw, circle] (d) [above right=of c] {$d$};
        \node[draw, circle] (d0) [below of=d] {$d_0$};
        \node[draw, circle] (y) [right=of d] {$y$};

        \path[-]
            (x) edge (a)
            (a) edge (a0)
            (a) edge (b)
            (b) edge (c)
            (c) edge (c0)
            (b) edge (f)
            (c) edge (d)
            (d) edge (d0)
            (f) edge (d)
            (d) edge (y)
        ;

    \end{tikzpicture}

    \caption{The gadget for each directed edge $(x,y)$.}
    \label{fig:directedEdgeGadget}
\end{center}\end{figure}
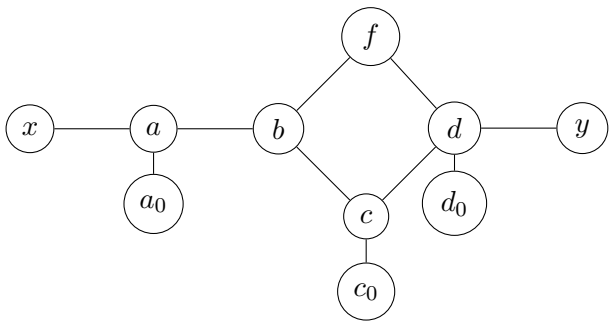

The correctness of the reduction hinges on this gadget acting like a directed edge from $x$ to $y$.  We assert that with two lemmas.  In this section, we provide proofs using nimber-based calculations.  In Appendix \ref{sec:alternateProofs}, we provide alternate proofs based on the winnability when added to $\ast$.  The reader is welcome to use whichever reasoning they prefer to follow.

We will also use some notation to represent a graph, $G$, after moves have been made.  For a subset $S\subset V$, we will use $G_{S}$ be the graph obtained from $G$ by removing $S$ and all edges incident to $S$.  Thus, from a position $(G, s)$, if a player chooses to traverse edge $(s, t)$, the resulting position is $(G_{\{s\}}, t)$.  Since we are often considering sets $S$ with just one vertex, we will shortcut examples like that by dropping the braces, as $(G_s, t)$.

\begin{lemma}[Wrong Way]
    Moving from $y$ to any vertex $d$ results in a value of $\ast 2$ or $\ast 3$.  In other words, $(G'_y, d) = \ast 2$ or $\ast 3$.
    \label{lem:wrongWay}
\end{lemma}

%We provide two versions of the proof.  The first, using nimbers, follows below.  The other version, characterized by the winnability of the position when added to a separate $\ast$-valued component, is given in the appendix, at lemma \ref{lem:wrongWayAlt}.  The reader is welcome to follow whichever reasoning they are more comfortable with.

\begin{proof}
    We prove this by examining the three options from $d$.  Moving to $d_0$ is clearly a move to 0.  Moving to $c$ will be non-zero, because $c_0$ is zero.  It remains to show that moving $d \rightarrow f$ results in a $\ast$-position.
    We can see this by considering the following necessary move $f \rightarrow b$.  Since both of $b$'s remaining neighbors, $a$ and $c$, have terminal neighbors ($a_0$ and $c_0$), they are non-zero.  Thus, the move to $b$ must be a zero position, and the move to $f$ must be $\ast$.
\end{proof}

\begin{lemma}[Correct Way]
    Moving from $d$ to $y$ results in a value of $\ast$ exactly when moving from $x$ to $a$ in the same gadget results in $\ast$.
    \label{lem:rightWay}
\end{lemma}

%We again provide two proofs, with the nimber version here and the adding-to-$\ast$ version in the appendix at lemma \ref{lem:rightWayAlt}.

\begin{proof}
    In both cases, we will use the fact that moving from $b$ to $f$ results in a 0-position, because $d$ is always non-zero and it is $f$'s only neighbor.

    For the first case, assume that moving $d$ to $y$ results in $\ast$.  This means that moving $c$ to $d$ has value $\ast 2$, as options to $f$ and $d_0$ are both 0.  Thus, $b$ to $c$ has value $\ast$.  Since $b$ has options to both 0 and $\ast$, moving to $b$ has value $\ast 2$, and moving $x$ to $a$ has value $\ast$.

    In the other case, assume that moving $d$ to $y$ does not have value $\ast$.  (Either it is $\ast 2$ or above or $y$ has already been removed.)  Thus, moving $c$ to $d$ results in a value of $\ast$, because $d$'s other options are 0.  $b$ to $c$ then has a value of $\ast 2$, meaning that $a$ to $b$ has a value of $\ast$.  This means that moving $x$ to $a$ has a value of $\ast 2 \neq \ast$, completing the proof.
\end{proof}

With our reduction and our two lemmas, we can now prove Theorem \ref{thm:nimberHardness}.  %The adding-to-$\ast$-version of the proof is included in the appendix at theorem \ref{thm:nimberHardnessAlt}.

\begin{proof}
    Determining the winnability of \rsGeog\ position $(G, s)$ is \cclass{PSPACE}-hard.  Thus, it remains to be shown that for the \rsUndirGeog\ position resulting from the reduction,
$(G', s)
\begin{cases}
    = \ast, &\text{ if } (G, s) = 0\ (\in \outP)\\
    \in \outN \setminus \{\ast\}, &\text{ if } (G,s) \in \outN
\end{cases}$

    Consider any \rsGeog\ position $(H, t)$, $k$ moves after $(G, x)$ and the analagous \rsUndirGeog\ position $(H', t) = r(H, t)$, reached $5k$ moves after $(G', s)$ by traversing the gadgets corresponding to the directed edges traversed to reach $(H, t)$.  If there are no options from $(H, t)$, then $(H', t)$ has options to $t_0$, which has value 0; possibly to gadget vertices $d$, which have value either $\ast 2$ or $\ast 3$ by Lemma \ref{lem:wrongWay}; and possibly to gadget vertices $a$ where the corresponding $y$ vertex has already been removed, which have a non-$\ast$ value (specifically $\ast 2$) by Lemma \ref{lem:rightWay}.

    Thus, there is a move to zero and might be moves to $\ast 2$ or $\ast 3$.  $(H', t) = \ast$.
  If there are options from $(H, t)$, then it is either in $\outP$ or $\outN$.  We can complete our proof inductively by assuming that the theorem is true for all options of $(H, t)$ and showing that it works for $(H, t)$.
 (1)
 If $(H, t) \in \outP$, then each option, $(H_t, p)$ is in $\outN$.  Thus, by our induction hypothesis, $(H_t', p) = \ast z$, where $z \geq 2$.  This means that $(H', t)$ doesn't have any options equal to $\ast$.  Since it does have a move to zero ($t_0$), $(H', t) = \ast$. \checkmark
   (2)
    If $(H, t) \in \outN$, then some option, $(H_t, p) \in \outP$. Thus, $(H_t', p) = \ast$.  Since $(H', t)$ has a move to zero ($t_0$) and $\ast$, the value is $\ast z$ where $z \geq 2$.
\end{proof}

Currently our reduction creates positions where the initial value (at $(G', s)$) is either $\ast$ or another nimber higher than 1.  We can narrow this down so that the question is whether we can distinguish between $\ast$ and $\ast 2$, specifically, by appending a Prelude gadget (see Figure \ref{fig:prelude}) before $s$ and then asking what the value of the overall game position is when starting at the ``start'' vertex.

\begin{figure}[h!]\begin{center}
    \begin{tikzpicture}[node distance = 1cm]
        \node[draw, circle] (start) {start};
        \node[draw, circle] (start0) [below of=start] {};
        \node[draw, circle] (start2) [right=of start] {};
        \node[draw, circle] (start20) [below of=start2] {};
        \node[draw, circle] (s) [right=of start2] {$s$};

        \path[-]
            (start) edge (start0)
            (start) edge (start2)
            (start2) edge (start20)
            (start2) edge (s)
        ;
    \end{tikzpicture}
    \caption{Prelude Gadget.}
       \label{fig:prelude}
\end{center}\end{figure}
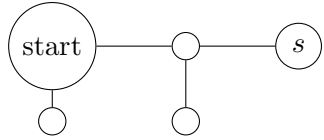

%\begin{corollary}
%    Determining whether an \rsUndirGeog\ position equals $\ast$ or $\ast 2$ is \cclass{PSPACE}-complete.
%\end{corollary}

\begin{corollary}
   Determining whether an \rsUndirGeog\ position equals $\ast$ or $\ast 2$ is \cclass{PSPACE}-complete, even on bipartite, planar graphs with a maximum degree of four.
\end{corollary}

\begin{proof}
    Since the height of the game tree is at most $n$, the Grundy-value can be computed in polynomial-space using the standard DFS technique \cite{PapadimitriouBook:1994}.

    Adding our prelude gadget to the reduction from Theorem \ref{thm:nimberHardness}, the value at vertex ``start'' will either be $\ast$ or $\ast 2$.  By calculating the nimber backtracking from the value at $s$, we see that it will be $\ast$ exactly when the value at $s$ is $\ast$, and $\ast 2$ for any of the other values of $s$.
%\end{proof}

%We notice that our gadgets have some extra properties: they preserve planarity of the original \rsGeog\ directed graph and they are bipartite.
%Together, they lead us to the corollary of our main theorem.
From Lichtenstein and Sipser \cite{LichtensteinSipser:1980}, we know that the winnability of \rsGeog\ is \cclass{PSPACE}-hard on bipartite, planar graphs with a maximum degree of three.  Our reduction preserves the planarity, and, since there is no odd-cycle in the gadget, the bipartite property as well.  We increase the degree by one because we add an extra vertex adjacent to the original vertices in $V$.  Thus, we are still hard on graphs with a maximum degree of four.
\end{proof}

The proof above has also established that
it is \cclass{PSPACE}-hard to determining whether two
``Fuzzy'' games $(G_1,s_1)$ and $(G_2,s_2)$
in \ruleset{Undirected Geography} has the same Grundy value.
In Section \ref{Sec:HardClassifier}, we will use our ``poly-high nimber constructor'' to prove the following theorem, showing no polynomial-time classifier exists the Grundy values in \ruleset{Undirected Geography}, beyond the well-known ``Zero''-``Fuzzy'' classifier, unless \cclass{P} = \cclass{PSPACE.}

\begin{theorem}[Too ``Fuzzy'' to Classify] \label{Thm:Mys}
In \ruleset{Undirected Geography}, determining whether or not the Grundy value of $(G,s)$, where $s$ has degree $\Delta$, is in any given set $S\subset [\Delta]$,
is \cclass{PSPACE}-hard.
%Making any distinction of the nim value of \ruleset{Undirected Geography}, beyond whether or not it is 0, is PSPACE-hard.
\label{thm:nimSeperation}
\end{theorem}

%\section{Dichotomy: The Charm of Three's Company}
\subsection{Following the Winning Way - A Polynomial-Time Branch-and-Bound}
\label{sec:BranchAndBound}
%In this section, we prove a dichotomy theorem,  highlighting the ``phase-transition to intractability'' of degree-four  \ruleset{Undirected Geography}.
We first show
that for any undirected graph $G$ with maximum-degree at most three,  the Grundy value of \ruleset{Undirected Geography} %over $G$
is polynomial-time computable.
In this case, we present a polynomial-time reduction from the Grundy-value computation to decision of winnability.
More broadly,  in Appendix \ref{sec:Abstract}, we generalize this finding of nimber-winnability reduction  abstractly to impartial games characterized by their {\em degrees} and {\em heights}.
In Appendix \ref{Appendix:Extensions}, we strengthen the result by analyzing the impact of high-degree nodes in the process.

For \ruleset{Undirected Geography} at a position $G=(V,E)$ and $s\in V$,
the {\em degree} of $s$ in $G$ is equal to the number of its feasible moves, and hence serves as a tight upper bound on the Grundy value of the position.
Similarly, the maximum  degree in $G$ characterizes the maximum {\em branching factor} of the game tree at position $(G,s)$:
If the maximum degree of $G$ is $\Delta$, then the branching factor of every node except the root is at most $\Delta-1$,
because the current geography path entering the node will take away at least one edge incident to the node.
The root may have branching factor $\Delta$ but no more.

\begin{theorem}[Following the Winning Way in \ruleset{Undirected Geography}]\label{Thm:BAB}
For any undirected graph $G=(V,E)$ with maximum degree $3$, and node $s\in V$,
the Grundy value at the \ruleset{Undirected Geography} position $(G,s)$
can be computed in polynomial time in $n=|V|$.
\end{theorem}
\begin{proof}
We first consider the case when the degree of $s$ is $1$ or $2$.
We then extend to the case of $3$.

\vspace{0.1in}
\noindent{\sc Single Option}:
When the degree of $s$ is $1$, the Grundy value of \ruleset{Undirected Geography} at $(G,s)$ is $\ast$ if and only if $(G,s)$ is a winning position.
So, the Grundy-value  can be directly reduced to the decision of winnability, which has a polynomial-time matching-based solution \cite{DBLP:journals/tcs/FraenkelSU93}.

\vspace{0.1in}
\noindent{\sc Double Options}:
When the degree of $s$ is $2$ (say with neighbors $v_1$ and $v_2$), the maximum branching factor of the game tree for position $(G,s)$ is $2$.
Let $G_{s}$ denote the graph obtained from $G$ by removing $s$ and edges incidents to $s$.
Then, the position with move to $v_1$ is $(G_{s},v_1)$
and with move to $v_2$ is $(G_{s},v_2)$.
Note that the degree of $v_1$ and $v_2$ in $G_{s}$ is at most $2$.
We run the polynomial-time matching-based winnability algorithm to determine whether or not
$(G_{s},v_1)$ and $(G_{s},v_2)$ are winning positions in \ruleset{Undirected Geography}, and consider
the  four cases:

\begin{enumerate}
\item {\bf [``Fuzzy'', ``Fuzzy'']} - both $(G_{s},v_1)$ and $(G_{s},v_2)$ are winning positions:
%\begin{eqnarray}
$\nimber(G,s) = 0$.
%\end{eqnarray}

\item {\bf [``Zero'', ``Zero'']} - both $(G_{s},v_1)$ and $(G_{s},v_2)$ are losing positions:
%\begin{eqnarray}
$\nimber(G,s) = \ast$.
%\end{eqnarray}

\item {\bf [``Fuzzy'', ``Zero'']} - $(G_{s},v_1)$ is a winning position and $(G_{s},v_2)$ is a losing position:

%\begin{eqnarray}
$\nimber(G,s) = \ast (3 - \nimber(G_{s},v_1))$
%\end{eqnarray}
\item {\bf [``Zero'', ``Fuzzy'']} -  $(G_{s},v_1)$ is a losing position and $(G_{s},v_2)$ is a winning position:

%\begin{eqnarray}
$\nimber(G,s) = \ast (3 - \nimber(G_{s},v_2))$
%\end{eqnarray}
\end{enumerate}

In the last two cases, one of $v_1$ and $v_2$ is 0, and the other has value $x = \ast$ or $x = \ast2$.  By the mex rule, $(G, s)$ will be the other of those values ($\ast 2$ or $\ast$, respectively) which is exactly $3 - x$ (or $3 \oplus x$), so the above derivation works.

In the first two cases, we have found
   the Grundy value of $(G,s)$ in polynomial time.
Crucial to the tractability, in the last two cases, we reduce the Grundy-value computation of $(G,s)$ to a single Grundy-value computation of either
$(G_{s},v_1)$ or $(G_{s},v_2)$.
Because $G_{s}$ has one less node than $G$,
the depth of the branch-and-bound process is $O(n)$.
In total, we made $O(n)$ calls to the
decision-of-winnability algorithm in order to compute the Grundy value of position $(G,s)$.

\vspace{0.1in}
\noindent{\sc Three Positions}: Finally, we consider the case when the degree of $s$ is three, say with
$u_1$, $u_2$, and $u_3$ as neighbors.
The maximum degree in $G_{s}$ remains at most three.
Then:
$$\nimber((G,s)) = \mex\left(
\left\{\nimber((G_{s},u_1),\nimber((G_{s},u_2),\nimber((G_{s},u_3)\right\}
\right)$$
Note that $u_i$ has degree at most two in $G_{s}$, $\forall i\in \{1,2,3\}$.
Thus, the nimbers of $(G_{s},u_i)$ can be computed, by our method above for degree one or two, in polynomial time.
%Then,
%From those three values, we then compute
%in a constant number of operations.
\end{proof}

\section{Games of Games, Sprague-Grundy Characterization, and \\ Mathematical-Computational Divergence} \label{sec:GoG}

Sprague-Grundy theory provides not only  a unified
theory for understanding diverse impartial rulesets, but also an elegant framework for their interactions.
Because for all impartial games $G, H$,
$\nimber(G+H) = \nimber(G) \oplus \nimber(H)$,
the Gurndy value of $(G+H)$
can be reduced in polynomial-time to the Grundy values of  $G$ and $H$.
In contrast,  using our complexity result for \ruleset{Undirected Geography}, we strengthen
Morris' theorem \cite{morris1981playing} from partisan to impartial games:

\begin{theorem}[Beyond Winning Impartial Games]\label{theo:WinningDecision}
If \cclass{P} $\neq $ \cclass{PSPACE}, then the disjunctive sum of two polynomial-time tractable impartial games can be intractable.
\end{theorem}
Thus, unlike Grundy-value computation,  there is no general polynomial-time reduction algorithm from winnability of $(G+H)$ to strategic determination for $G$ and $H$, unless
\cclass{P} = \cclass{PSPACE}.
This illustrates a striking view of the classical Sprague-Grundy characterization through the lens of computational complexity theory.
The Grundy value and strategic determination are two different yet fundamental summaries of the game tree.
Using the complexity gulf between \cclass{P} and \cclass{PSPACE}, our results help to demonstrate that the Grundy value is a significantly richer summary of game data than strategic determination. In fact, our results have established the following:
\begin{theorem}[Intractability of Game Summary]\label{Thm:Summary}
Unless \cclass{P} = \cclass{PSPACE},
there is no general polynomial-time method to summarize
 two given impartial games (say $G$ and $H$) to efficiently solve the game of their sum ($G+H$).
\end{theorem}

Sprague-Grundy theory established that such concise summaries---in the form of Grundy values---of game data always exist
when computational cost is no object.
This work highlights a subtle yet fundamental contrast between the
 mathematical facet and computational facet of combinatorial game theory.
Applying Sprague-Grundy theory,
 our dichotomy theorem %on Grundy-value
also enables us to settle the solvability of several families of games extending \ruleset{Undirected Geography}.

\begin{itemize}
\item \ruleset{Multi-Field Undirected Geography} - The disjunctive sum of two or more
of \ruleset{Undirected Geography} games.
\item \ruleset{Multi-Token Undirected Geography} -
Like \ruleset{Undirected Geography},
%\ruleset{Multi-Token Undirected Geography}
this game is played on an undirected graph, in which a game position is defined by a graph $G =(V,E)$ and a set $S\subset V$.
Each node in $S$ has a token, and in each turn, exactly one of the tokens can be moved to an adjacent {\em unoccupied node}, and the node of its previous location is removed from the graph.
In \ruleset{Multi-Token Undirected Geography}, alternating moves by two players  create multiple node-disjoint exploring paths in $G$,
one by each token.
The game ends when no valid extension exists to any of these paths.

\item \ruleset{Undirected Geography with Passes} - Another natural extension of a game is to allow players to pass their turn.
Here, for an integer $k$, we consider
\ruleset{Undirected Geography}  with $k$-Total Passes, which is an \ruleset{Undirected Geography} game
 whose feasible moves are augmented by allowing players to pass their turn, provided that the total number of passes taken so far
 (by both players) is less than $k$.

\item \ruleset{Swap Uno} - This game is inspired by a generalization of Uno,  was shown by Domaine {\em et al} \cite{demaine2014uno} to be in P via reduction to \ruleset{Undirected Geography}. In this game, there are two hands, $H_1$ and $H_2$, which each consist of a set of cards. This is a perfect information game, so both players may see each other's hands. Each card has two attributes, a color $c$ and a rank $r$ and can be represented as $(c, r)$. A card can only be played in the center (shared) pile if the previous card matches either the $c$ of the current card or the $r$ of the current card. Finally, for the special part that makes this ''Swap'' Uno, either player may, once a game, decide to use their turn to swap their hand for their opponent's rather than playing in a pile. Once a single player swaps, the other player may not swap.

%As such, we can imagine a simple extension of that, where each player has 2 hands, and there are two piles. Either player can either play their left hand into the left pile or their right hand into the right pile. This game, in contrast to 2 Player Uno, which is in P, is PSPACE-complete! %need to make this fit in better
\end{itemize}

We now prove that, %in spite of the fact that
although \ruleset{Undirected Geography} is polynomial-time solvable, these basic extensions of \ruleset{Undirected Geography} can be more challenging computationally.

\begin{theorem}[The War of \ruleset{Geography} Battles]\label{Theo:GameOfGames}
 Deciding whether or not the current player has a winning strategy
 in the sum of two \ruleset{Undirected Geography} games, and consequently, in  \ruleset{Two-Token Undirected Geography} and \ruleset{Swap Uno} games, is \cclass{PSPACE}-complete.
Furthermore, \ruleset{Undirected Geography with $k$-Total Passes} is \cclass{PSPACE}-complete to solve for {\em odd} $k$, and polynomial-time solvable for {\em even} $k$.
\end{theorem}
\begin{proof}
Let's start with the complexity analysis for solving \ruleset{Undirected Geography with $k$-Total Passes}.
We first consider a trivial game, to be referred to as \ruleset{Pass}.
Each position in \ruleset{Pass} is defined by an integer $k$.
The terminal position is the one with $k=0$.
For any $k>0$, there is a single move at position $k$ to position
$k-1$.
\ruleset{Pass} with $k=1$ is isomorphic to \ruleset{Nim} with a single pile of one item.
In general, the \ruleset{Pass} position $k$ is isomorphic
to \ruleset{Nim} with $k$ piles, each  containing a single item.
%By the Sprague-Grundy Theorem,
The Grundy value of \ruleset{Pass} at position $k$ is zero if $k$ is even and $\ast$ if $k$ is odd.
For any undirected graph $G$ and positive integer $k$, \ruleset{Undirected Geography with $k$-Total Passes} at position $((G,s),k)$ is isomorphic to
the game  defined by the disjunctive sum of two battlefield games:
(1) \ruleset{Undirected Geography} at position $(G,s)$ and (2) \ruleset{Pass} at position $k$.
Therefore:
\begin{itemize}
\item
When $k$ is odd (e.g., $k=1$),  by the Sprague-Grundy theory,
the Grundy value of position $((G,s),k)$
%\ruleset{Undirected Geography with $k$-Total Passes}
is equal to  $\nimber((G,s))  \oplus \ast$.
Thus, the current player in this game has NO winning strategy if and only if $\nimber((G,s)) = \ast$ in \ruleset{Undirected Geography}.
We conclude that the winnability of this game is \cclass{PSPACE}-complete to solve,
  because deciding whether or not $\nimber((G,s))$ is equal to $\ast$ (or $\ast 2$) is \cclass{PSPACE}-complete (Theorem \ref{Thm:Mys}).

\item When $k$ is even,  the Grundy value of position $((G,s),k)$ is equal to $\nimber((G,s))$, for which we can
distinguish ``Zero'' from ``Fuzzy''
in  polynomial time.
\end{itemize}

We can similarly characterize the complexity of the sum of two \ruleset{Undirected Geography} games:
Let $G'$ be the two-node graph with edge $(s',u')$, and let $G = (V,E)$ be an arbitrary undirected graph with $s\in V$.
Then, by Sprague-Grundy theory and the fact that $\nimber((G',s'))= \ast$, we have:
$\nimber((G,s)+(G',s')) = \nimber((G,s)) \oplus  \ast$.
Thus, determining the winnability of the disjunctive sum of two \ruleset{Undirected Geography} games is \cclass{PSPACE}-complete.

Because the disjunctive sum of two \ruleset{Undirected Geography} games is a special case of \ruleset{Two-Token Undirected Geography},
the \cclass{PSPACE}-hardness extends,
 and, in fact, even when we require that the underlying graph is connected.

Finally, \ruleset{Swap Uno} is the disjunctive sum of
\ruleset{Uncooperative Uno} and *.
In \cite{demaine2014uno}, Demaine {\em et al}  gave a simple reduction from \ruleset{Uncooperative Uno} to \ruleset{Undirected Geography} over bipartite graphs that isomorphically preserves players' options, setting up their polynomial-time solution.
In Appendix \ref{Appendix:Uno}, we then show that
\ruleset{Uno} bipartite graphs have the structural property
needed to encode the hard instances for Grundy value computation in \ruleset{Undirected Geography},
as required in our proof for Theorem \ref{thm:nimberHardness}.
Consequently,  the \cclass{PSPACE}-hardness of \ruleset{Swap Uno} follows from  that of \ruleset{Undirected Geography with One Pass}.
\end{proof}

\begin{proof}(of Theorems \ref{theo:WinningDecision}and \ref{Thm:Summary})
Both follow directly from Theorem \ref{Theo:GameOfGames} on the \cclass{PSPACE}-hardness of the disjunctive sum of two \ruleset{Undirected Geography} games, and
the sum of \ruleset{Nim} and \ruleset{Undirected Geography}.
\end{proof}

\begin{corollary}
Even in a sum where we fix the second game, so long as that game is ``Fuzzy'', the problem is still intractable.
\end{corollary}
\begin{proof}
Using Theorem \ref{thm:nimSeperation}, we can simply have an $S$ that is just the Grundy value of the second game, which means that the problem is a winning one if and only if the nimber is not in $S$.
\end{proof}

On the tractable side, it follows from Theorem \ref{Thm:BAB} that:

\begin{corollary}
The sum of any collection of \ruleset{Undirected Geography} games over degree-three graphs are polynomial-time solvable.
\end{corollary}

\section{\ruleset{Undirected Geography} with Polynomial Grundy Values} \label{sec:Math}

A fundamental problem in combinatorial game theory is that of \textit{nimber constructability}.
That is to say---when specialized to the game of our focus---the question of whether a game of \ruleset{Undirected Geography} can actually have a certain Grundy value (equivalent to determining the habitat for impartial games), and if it can, whether it can be succinctly encoded.
The existence is important primarily from a pure mathematical standpoint. The succinct encoding is needed for sums of games with high Grundy values to actually be shown intractable.
In Section \ref{Sec:HardClassifier}, we prove Theorem \ref{thm:nimSeperation} using the {\em (polynomially) succinct encoding} of high nimbers to support our complexity analyses.

\subsection{Logarithmic Intuition and Polynomial Challenge}

The habitat going up to the maximum degree in the graph is simple.
We will present it in the next construction  %as we use the fact we can do this in
to motivate our more advanced proof.

\begin{observation}[Logorithmic Nimber in Trees]
Through a simple tree structure, one may create an \ruleset{Undirected Geography} position with nimber $\ast n$, highest degree $n$, and $2^n$ vertices.
\label{obs:tree}
\end{observation}
\begin{proof}
We can just recursively define a tree $t(n)$ which has moves to $t(n-1), t(n-2) \dots t(0)$. For the base case, $t(0)$ is a single isolated vertex, which clearly has Grundy value 0. The size of this tree is $2^0$ for the base case, and we can inductively assume each of the smaller $t(i)$ have $2^i$ vertices, so we have $2^{k+1}$ be $2^0 + 2^1 + \dots + 2^k + 1$, where the final 1 is the new root.
\end{proof}

Thus, we can get a poly-log nimber using a polynomial number of vertices (and certainly any constant nimber, which we will use for gadgets up to $\ast 3)$. %However, all we will use from this for the stronger proof is that this lets us construct trees for constant nimbers, namely up to $\ast 3$.

%The more important thing for us here is the intuition this provides for the stronger version.
The exponential size of the basic construction comes from the fact that we repeat each tree in each subtree. This is necessary, since if we attempt to combine the subtrees, being able to move "back up" those trees could change the Grundy values. %will not be preserved (as there will be extra moves "up" which can change the value).
Note here that in \ruleset{Generalized Geography}, one can use directed edges to prevent undesired ``up'' moves to share the lower nimber nodes.
Thus, it is in fact straightforward to achieve nimber $n$ with $n+1$ vertices.
We can't just replace these with our directed-edge gadget from Figure
\ref{fig:directedEdgeGadget}, because the inner degree on those is constant and will prevent us to get arbitrarily large nimbers.  We need a more sophisticated mechanism to get nimbers of any size.

\begin{figure}[h!]\begin{center}
    \scalebox{.9}{
    \begin{tikzpicture}[node distance = 2cm and 1.2cm]

        \node[draw, rectangle] (Nn) {$N_n$};
        \node[draw, rectangle] (star1n) [above left=of Nn] {$\ast$};
        \node[draw, rectangle] (star2n) [below left=of Nn] {$\ast 2$};
        \node[draw, rectangle] (Nnm1) [below left=of star1n] {$N_{n-1}$};
        \node[draw, rectangle] (star1nm1) [above left=of Nnm1] {$\ast$};
        \node[draw, rectangle] (star2nm1) [below left=of Nnm1] {$\ast 2$};
        \node[rectangle] (dots) [below left=of star1nm1] {$\cdots$};
        \node[draw, rectangle] (star17) [above left=of dots] {$\ast$};
        \node[draw, rectangle] (N6) [below left=of star17] {$N_5$};
        \node[draw, rectangle] (star27) [below right=of N6] {$\ast 2$};
        \node[draw, rectangle] (star16) [above left=of N6] {$\ast$};
        \node[draw, rectangle] (N5) [below left=of star16] {$N_4$};
        \node[draw, rectangle] (star26) [below right=of N5] {$\ast 2$};

        \path[-]
            (Nn) edge (star1n)
            (Nn) edge (star2n)
            (Nn) edge (Nnm1)
            (Nnm1) edge (star1nm1)
            (Nnm1) edge (star1n)
            (Nnm1) edge (star2nm1)
            (Nnm1) edge (star2n)
            (N6) edge [bend left=10] (Nnm1)
            (N6) edge [bend left=15] (Nn)
            (N6) edge (star17)
            (N6) edge (star27)
            (N6) edge (star1n)
            (N6) edge (star2n)
            (N6) edge (star1nm1)
            (N6) edge (star2nm1)
            (N6) edge (star16)
            (N6) edge (star26)
            (N5) edge (star16)
            (N5) edge (star26)
            (N5) edge (N6)
            (N5) edge [bend right=15] (Nnm1)
            (N5) edge (star17)
            (N5) edge (star27)
            (N5) edge (star1nm1)
            (N5) edge (star2nm1)
            (N5) edge [bend left = 15] (Nn)
            (N5) edge (star1n)
            (N5) edge (star2n)
        ;
    \end{tikzpicture}}
    \caption{The overall schema for Nimber Constructability.  The value of the position with the token at $N_n$ is $\ast n$.}
    \label{fig:nimberConstructibilitySchema}
\end{center}\end{figure}
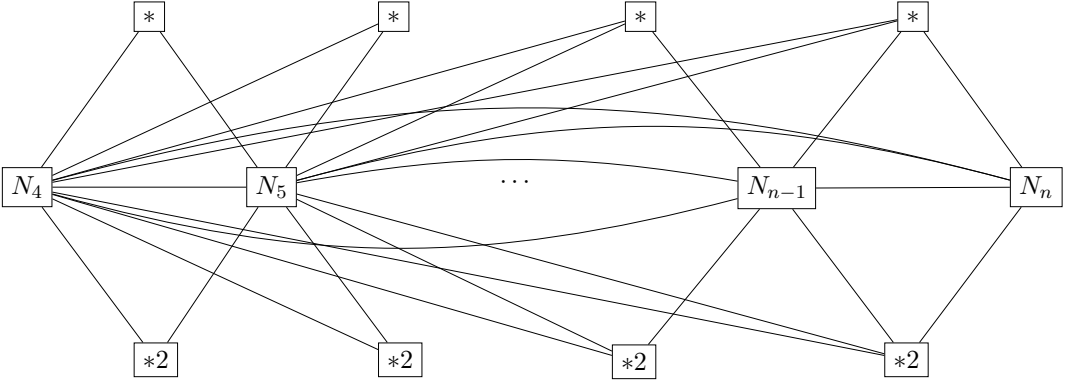

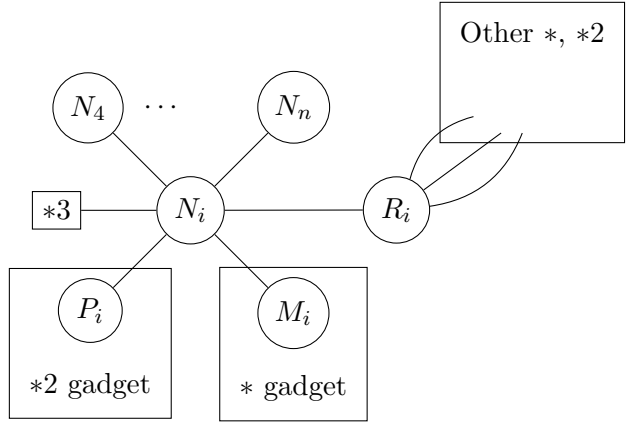
\begin{figure}[h!]\begin{center}
    \begin{tikzpicture}[node distance = 1cm]
        \node[draw, circle] (Ni) {$N_i$};
        \node[draw, circle] (N5) [above left=of Ni] {$N_4$};
        \node[] (dotsN) [right of=N5] {$\cdots$};
        \node[draw, circle] (Nn) [above right=of Ni] {$N_n$};
        \node[draw, rectangle] (star3) [left=of Ni] {$\ast 3$};
        \node[draw, circle] (Ri) [below right=of Nn] {$R_i$};
        \node[] (higherStars) [above right=of Ri] {\phantom{monkey}};
        \node[] (hiStarsLabel) [above of=higherStars] {Other $\ast$, $\ast 2$};
        \node[draw, rectangle, fit=(higherStars) (hiStarsLabel)] {};
        \node[draw, circle] (Pk) [below left=of Ni] {$P_i$};
        \node[] (PkLabel) [below of=Pk] {$\ast 2$ gadget};
        \node[draw, rectangle, fit=(Pk) (PkLabel)] {};
        \node[draw, circle] (Mk) [below right=of Ni] {$M_i$};
        \node[] (MkLabel) [below of=Mk] {$\ast$ gadget};
        \node[draw, rectangle, fit=(Mk) (MkLabel)] {};

        \path[-]
            (Ni) edge (N5)
            (Ni) edge (Nn)
            (Ni) edge (star3)
            (Ni) edge (Ri)
            (Ri) edge [bend left] (higherStars)
            (Ri) edge (higherStars)
            (Ri) edge [bend right] (higherStars)
            (Ni) edge (Pk)
            (Ni) edge (Mk)
        ;
    \end{tikzpicture}
    \caption{Each vertex $N_i$ is connected to all other $N$-vertices, as well as its own $\ast 3$ gadget, it's own $\ast 2$ gadget, it's own $\ast$ gadget, and it's own $R_i$ vertex (with value 0).  $R_i$ is also connected to the $P_k$ and $M_k$ gadgets where $k > i$ as shown in the following figures.}
    \label{fig:Ni}
\end{center}\end{figure}

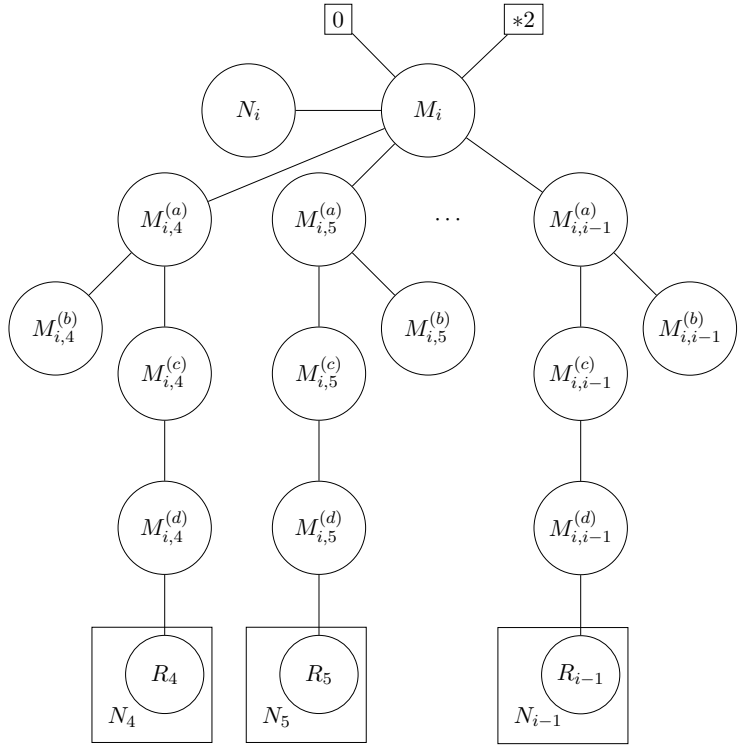
\begin{figure}[h!]\begin{center}
    \scalebox{.8}{
    \begin{tikzpicture}[node distance = 1cm, gNode/.style = {circle, minimum width = 1.55cm}]
        \node[draw, gNode] (Mi) {$M_i$};
        \node[draw, rectangle] (zero) [above left=of Mi] {0};
        \node[draw, rectangle] (star2) [above right=of Mi] {$\ast 2$};
        \node[draw, gNode] (Ni) [below left=of zero] {$N_i$};
        \node[draw, gNode] (Mi5a) [below left=of Mi] {$M_{i,5}^{(a)}$};
        \node[draw, gNode] (Mi4a) [left=of Mi5a] {$M_{i,4}^{(a)}$};
        \node[draw, gNode] (Mi4b) [below left=of Mi4a] {$M_{i,4}^{(b)}$};
        \node[draw, gNode] (Mi4c) [below=of Mi4a] {$M_{i,4}^{(c)}$};
        \node[draw, gNode] (Mi4d) [below=of Mi4c] {$M_{i,4}^{(d)}$};
        \node[draw, circle, minimum width = 1.3cm] (R4) [below=of Mi4d] {$R_4$};
        \node[] (N4Label) [below left of=R4] {$N_4$};
        \node[draw, rectangle, fit=(R4) (N4Label)] {};

        \node[draw, gNode] (Mi5b) [below right=of Mi5a] {$M_{i,5}^{(b)}$};
        \node[draw, gNode] (Mi5c) [below=of Mi5a] {$M_{i,5}^{(c)}$};
        \node[draw, gNode] (Mi5d) [below=of Mi5c] {$M_{i,5}^{(d)}$};
        \node[draw, circle, minimum width = 1.3cm] (R5) [below=of Mi5d] {$R_5$};
        \node[] (N5Label) [below left of=R5] {$N_5$};
        \node[draw, rectangle, fit=(R5) (N5Label)] {};

        \node[] (dots) [right=of Mi5a] {$\cdots$};

        %k is i-1
        \node[draw, gNode] (Mika) [right=of dots] {$M_{i,i-1}^{(a)}$};
        \node[draw, gNode] (Mikb) [below right=of Mika] {$M_{i,i-1}^{(b)}$};
        \node[draw, gNode] (Mikc) [below=of Mika] {$M_{i,i-1}^{(c)}$};
        \node[draw, gNode] (Mikd) [below=of Mikc] {$M_{i,i-1}^{(d)}$};
        \node[draw, circle, minimum width = 1.3cm] (Rk) [below=of Mikd] {$R_{i-1}$};
        \node[] (NkLabel) [below left of=Rk] {$N_{i-1}$};
        \node[draw, rectangle, fit=(Rk) (NkLabel)] {};

        \path[-]
            (Mi) edge (zero)
            (Mi) edge (star2)
            (Mi) edge (Ni)
            (Mi) edge (Mi4a)
            (Mi4a) edge (Mi4b)
            (Mi4a) edge (Mi4c)
            (Mi4c) edge (Mi4d)
            (Mi4d) edge (R4)
            (Mi) edge (Mi5a)
            (Mi5a) edge (Mi5b)
            (Mi5a) edge (Mi5c)
            (Mi5c) edge (Mi5d)
            (Mi5d) edge (R5)
            (Mi) edge (Mika)
            (Mika) edge (Mikb)
            (Mika) edge (Mikc)
            (Mikc) edge (Mikd)
            (Mikd) edge (Rk)
        ;
    \end{tikzpicture}}
    \caption{A $\ast$ gadget, which has value $\ast$ unless a lower-rank-$N$ vertex is removed.  If all of $N_k$ vertices exist off the bottom, then moving to any one of them from $R_k$ results in a 0-board (by Lemma \ref{lem:rToN}).  This causes the move $M_i \rightarrow M_{i,k}^{(a)}$ to be equal to $\ast 2$.  Otherwise, one $R_k$ has value 0, so $M_i \rightarrow M_{i,k}^{(a)}$ is a move to $\ast$, so $M_i$ instead has value $\ast 3$.}
    \label{fig:Mi}
\end{center}\end{figure}

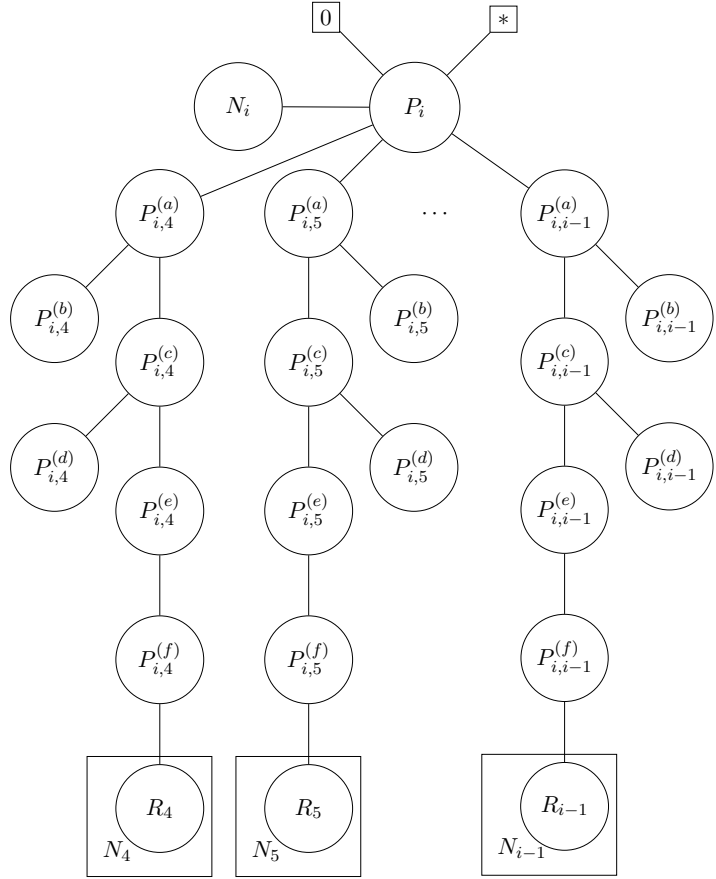
\begin{figure}[h!]\begin{center}
    \scalebox{.8}{
    \begin{tikzpicture}[node distance = 1cm, gNode/.style = {circle, minimum width = 1.46cm}]
        \node[draw, circle, minimum width = 1.5cm] (Pi) {$P_i$};
        \node[draw, rectangle] (zero) [above left=of Pi] {0};
        \node[draw, rectangle] (star) [above right=of Pi] {$\ast$};
        \node[draw, gNode] (Ni) [below left=of zero] {$N_i$};
        \node[draw, gNode] (Pi5a) [below left=of Pi] {$P_{i,5}^{(a)}$};
        \node[draw, gNode] (Pi4a) [left=of Pi5a] {$P_{i,4}^{(a)}$};
        \node[draw, gNode] (Pi4b) [below left=of Pi4a] {$P_{i,4}^{(b)}$};
        \node[draw, gNode] (Pi4c) [below=of Pi4a] {$P_{i,4}^{(c)}$};
        \node[draw, gNode] (Pi4d) [below left=of Pi4c] {$P_{i,4}^{(d)}$};
        \node[draw, gNode] (Pi4e) [below=of Pi4c] {$P_{i,4}^{(e)}$};
        \node[draw, gNode] (Pi4f) [below=of Pi4e] {$P_{i,4}^{(f)}$};
        \node[draw, gNode] (R4) [below=of Pi4f] {$R_4$};
        \node[] (N4Label) [below left of=R4] {$N_4$};
        \node[draw, rectangle, fit=(R4) (N4Label)] {};

        \node[draw, gNode] (Pi5b) [below right=of Pi5a] {$P_{i,5}^{(b)}$};
        \node[draw, gNode] (Pi5c) [below=of Pi5a] {$P_{i,5}^{(c)}$};
        \node[draw, gNode] (Pi5d) [below right=of Pi5c] {$P_{i,5}^{(d)}$};
        \node[draw, gNode] (Pi5e) [below=of Pi5c] {$P_{i,5}^{(e)}$};
        \node[draw, gNode] (Pi5f) [below=of Pi5e] {$P_{i,5}^{(f)}$};
        \node[draw, gNode] (R5) [below=of Pi5f] {$R_5$};
        \node[] (N5Label) [below left of=R5] {$N_5$};
        \node[draw, rectangle, fit=(R5) (N5Label)] {};

        \node[] (dots) [right=of Pi5a] {$\cdots$};

        %k is i-1
        \node[draw, circle] (Pika) [right=of dots] {$P_{i,i-1}^{(a)}$};
        \node[draw, circle] (Pikb) [below right=of Pika] {$P_{i,i-1}^{(b)}$};
        \node[draw, circle] (Pikc) [below=of Pika] {$P_{i,i-1}^{(c)}$};
        \node[draw, circle] (Pikd) [below right=of Pikc] {$P_{i,i-1}^{(d)}$};
        \node[draw, circle] (Pike) [below=of Pikc] {$P_{i,i-1}^{(e)}$};
        \node[draw, circle] (Pikf) [below=of Pike] {$P_{i,i-1}^{(f)}$};
        \node[draw, gNode] (Rk) [below=of Pikf] {$R_{i-1}$};
        \node[] (NkLabel) [below left of=Rk] {$N_{i-1}$};
        \node[draw, rectangle, fit=(Rk) (NkLabel)] {};

        \path[-]
            (Pi) edge (zero)
            (Pi) edge (star)
            (Pi) edge (Ni)
            (Pi) edge (Pi4a)
            (Pi4a) edge (Pi4b)
            (Pi4a) edge (Pi4c)
            (Pi4c) edge (Pi4d)
            (Pi4c) edge (Pi4e)
            (Pi4e) edge (Pi4f)
            (Pi4f) edge (R4)
            (Pi) edge (Pi5a)
            (Pi5a) edge (Pi5b)
            (Pi5a) edge (Pi5c)
            (Pi5c) edge (Pi5d)
            (Pi5c) edge (Pi5e)
            (Pi5e) edge (Pi5f)
            (Pi5f) edge (R5)
            (Pi) edge (Pika)
            (Pika) edge (Pikb)
            (Pika) edge (Pikc)
            (Pikc) edge (Pikd)
            (Pikc) edge (Pike)
            (Pike) edge (Pikf)
            (Pikf) edge (Rk)
        ;
    \end{tikzpicture}}
    \caption{The $\ast 2$ gadget}
    \label{fig:Pi}
\end{center}\end{figure}

\subsection{Polynomial-High Nimber Constructability}

To attain nimber $n$, we create $n$ vertices $N_4, \ldots, N_n$, which exist in a clique, as in Figure \ref{fig:nimberConstructibilitySchema}.  Each $N_i$ has Grundy value $\ast i$ so long as all $N_k$  with $k < i$ remain.  (We will say these vertices have a lower \textit{rank}.)  We argue that starting with the token on vertex $N_n$ is a $\ast n$-position.  (We do not have $N_0$, $N_1$, $N_2$, or $N_3$, since we use 0 through $\ast 3$ as mechanisms in our structure to ensure the player is unable to move ``up'' in rank.)

After any move from $N_n$ to $N_k$, we no longer want vertices with higher-rank than $k$ to retain their nimber value.  To attain this,
%To prevent increasing the nimber by moving to a vertex with a higher rank,
we create $\ast$ and $\ast 2$ gadgets for each $N_i$, which have their named values if and only if no vertex of a lower rank has been removed from the graph. As such, a later move to a higher-rank $N_i$ will have value
%This will make it so moving to a higher rank results in value
either $\ast$ or $\ast 2$ instead of $\ast i$.

We present these designs in Figures
\ref{fig:Ni}, \ref{fig:Mi}, and \ref{fig:Pi} for illustration. In addition to the figures, we also include a formal algorithmic formulation in Algorithm \ref{alg:nimberGeneration}:

\begin{algorithm}[h]
\If{$n \leq 3$}{
      Return a tree representation of $\ast n$ with the corresponding start position
   }
\Else{
       Let the graph we are working on be $G = (V,E)$

       Add vertex $N_4$ and $R_4$ to $V$

       \For{$i = 5$ to $n$}{
           Add vertices $N_i, M_i, P_i, R_i, M_i^{(0)}$, and $P_i^{(0)}$ to $V$

           Add edges $(N_i, R_i), (M_i, M_i^{(0)})$, and $(P_i, P_i^{(0)})$ to $E$

           Create tree versions of $\ast 3$ to $V$ and $E$, and connect $N_i$ to it

           \For{$j = 4$ to $i-1$}{
               Add vertices $P_{ij}^{(a)}, P_{ij}^{(b)}, P_{ij}^{(c)}, P_{ij}^{(d)}, P_{ij}^{(e)}, P_{ij}^{(f)}, M_{ij}^{(a)}, M_{ij}^{(b)}, M_{ij}^{(c)}$, and $M_{ij}^{(d)}$ to $V$

               Add edges $(P_i, P_{ij}^{(a)}), (P_{ij}^{(a)}, P_{ij}^{(b)}), (P_{ij}^{(a)},  P_{ij}^{(c)}), (P_{ij}^{(c)}, P_{ij}^{(d)})$, and $(P_{ij}^{(c)},  P_{ij}^{(e)})$ to $E$

               Add edges $(P_{ij}^{(e)}, P_{ij}^{(f)}) (P_{ij}^{(f)}, R_j)$, and $(N_i, N_j)$ to $E$

               $(M_i, M_{ij}^{(a)}), (M_{ij}^{(a)}, M_{ij}^{(b)}), (M_{ij}^{(a)}, M_{ij}^{(c)}), (M_{ij}^{(c)}, M_{ij}^{(d)})$ and $(M_{ij}^{(d)}, R_j)$ to $E$
           }
           Add a tree representation of $\ast 2$ to $V$ and $E$, and add an edge to $M_i$ to $E$

           Add a tree representations of $\ast$ to $V$ and $E$, and an edge connect this to $P_i$
       }
       Add the tree versions of $\ast$ and $\ast 2$ the edge to $N_4$ to $V$ and $E$

       Return $G$ and starting node $N_n$
}
\caption{Nimber Generation Algorithm}
\label{alg:nimberGeneration}
\end{algorithm}

\begin{observation}
The Nimber Generation Algorithm runs in polynomial time.
\end{observation}
\begin{proof}
By simple observation, each loop has no more than $n$ iterations, and there are never more than two nested loops. Each line in each loop runs in constant time. The running time, in fact, is $O(n^2)$.
\end{proof}

\begin{lemma}[Grounded]
In a game where the only vertices removed are some $N_i$ vertices along with some of their associated $M_i$ and $M_{ip}$ vertices or $P_i$ and $P_{ip}$ vertices, then traversing edge $(N_k, R_k)$ always results in a move to 0.
\label{lem:baseZero}
\end{lemma}
\begin{proof}
Any move to an $M_{i,k}^{(a)}$ vertex has nimber at least $\ast$ since it has a move to $M_{i,k}^{(b)}$ which is 0.  Thus, moving $M_{i,k}^{(c)}$ from $M_{i,k}^{(d)}$ results in a 0, so all $M_{i,k}^{(d)}$ moves from $R_k$ result in $\ast$.
The same is true of moving from $R_k$ to $P_{i,k}^{(f)}$ because $P_{i,k}^{(c)}$ is also non-zero.  Since $R_k$ only has $\ast$-options, it's value is zero when moving from $N_k$.
%When on $R_i$, there are a series of $M_{ij}^{(d)}$ and $P_{ij}^{(f)}$ moves, each of which have value $\ast$. For  which has a move to  Thus, $M_{ij}^{(c)}$ has value 0 and $M_{ij}^{(d)}$ value $\ast$.
%Moving to $P_{ij}^{(f)}$ has $\ast$, since $P_{ij}^{(e)}$ has value 0. since $P_{ij}^{(c)}$ has value at least $\ast$ since the move to $P_{ij}^{(d)}$ is a move to 0.
\end{proof}

\begin{lemma}
In a game where the only vertices removed are some $N_i$ vertices along with some of their associated $M_i$ and $M_{ip}$ vertices or $P_i$ and $P_{ip}$ vertices, then traversing the edge $(R_p, N_p)$ is a move to 0.
\label{lem:rToN}
\end{lemma}
\begin{proof}
There is a move to $\ast 3$ (and, if $p = 4$, $\ast 2$ and $\ast $), moves to $M_p$ and $P_p$, which both have a move to 0 by construction, and to various other $N_i$, which have moves to $R_i$, which are moves to 0 by Lemma \ref{lem:baseZero}.
\end{proof}

\begin{lemma}[Skip $*2$]
So long as only $N_i$ vertices are removed from the graph, the position from moving from $N_k$ to $M_k$ has nim-value $\ast$ if and only if no $N_i$ have been removed with $i < k$. Otherwise, it is $\ast 3$.
\label{lem:mValue}
\end{lemma}
\begin{proof}
Consider the result of moving $M^{(d)}_{i,k} \rightarrow R_k$.  All moves to other $M^{(d)}_{j,k}$ vertices are losing moves as established in Lemma \ref{lem:baseZero}.  There is only a winning move if $N_k$ still exists, so the position at $R_k \neq 0$ iff $N_k$ still exists.  Let's consider these two cases:

(1) If $N_k$ exists, then moving to $R_k \neq 0$.  Thus, moving to $M^{(d)}_{i,k}$ yields $0$, so moving to $M^{(c)}_{i,k}$ yields $\ast$, and moving to $M^{(a)}_{i,k}$ yields $\ast 2$.  If all $N_k$ exist, then $M_i$ does not have a move to $\ast$, so it's value is $\ast$ from $N_i$.

(2) On the other hand, if $N_k$ does not exist, then moving to $R_k$ from $M^{(d)}_{i,k}$ yields $0$.  Thus, moving to $M^{(d)}_{i,k}$ yields $\ast$, so moving to $M^{(c)}_{i,k}$ yields $0$, and moving to $M^{(a)}_{i,k}$ yields $\ast$.  Since $M_i$ has an $\ast$-option, moving there from $N_i$ now yields a position with value $\ast 3$.
%$M_k$, by construction, has a move to 0 and $\ast 2$. Each other move is to an $M_{kj}^{(a)}$. Each of these have a move to $M_{kj}^{(b)}$, which is always a 0, and to $M_{kj}^{(c)}$, which is has a move to $M_{kj}^{(d)}$, which has a move to $R_j$. This is 0 if $N_j$ has been removed from the graph, since as discussed in lemma \ref{lem:baseZero}, all the moves to $M_{ij}^{(d)}$ and $P_{ij}^{(f)}$ and have value $\ast$. If the $N_j$ remains, then that move is a move to 0 by lemma \ref{lem:rToN}, thus making this $\ast 2$. Thus, if $N_j$ has been removed, $M_{kj}^{(a)}$ will have value $\ast$, and otherwise will have value $\ast 2$. Thus, $M_k$ has value $\ast$ if and only if no $N_j$ has been removed, and otherwise has value $\ast 3$.
\end{proof}

\begin{lemma}[Skip $*$]
So long as only $N_i$ vertices are removed from the graph, a token on $P_k$ has nim-value $\ast 2$ if and only if no $N_i$ have been removed with $i < k$. If the value is not $\ast 2$, it is $\ast 3$.
\label{lem:pValue}
\end{lemma}
\begin{proof}
By the same logic as in Lemma \ref{lem:mValue}, we see that $R_k$ has value 0 exactly when $N_k$ no longer exists.  We examine the two cases to complete the proof:

(1) If $N_k$ exists, $R_k$ is non-zero, so moving to $P_{i,k}^{(f)}$ from above yields 0.  Working back up, moving to $P_{i,k}^{(e)}$ yields $\ast$, moving to $P_{i,k}^{(c)}$ yields $\ast 2$, and moving to $P_{i,k}^{(a)}$ yields $\ast$.  If these all exist, then $P_i$ has moves to only 0 and $\ast$, so it is $\ast 2$.

(2) If some $N_k$ doesn't exist, $R_k$ is zero, so moving to $P_{i,k}^{(f)}$ from above yields $\ast$.  Again, working back up, moving to $P_{i,k}^{(e)}$ yields 0, moving to $P_{i,k}^{(c)}$ yields $\ast$, and moving to $P_{i,k}^{(a)}$ yields $\ast 2$.  Now that $P_i$ has an $\ast 2$-option, it is instead $\ast 3$ from $N_i$.
%By construction, $P_k$ has a move to 0, $\ast$, and various $P_{kj}^{(a)}$. Each $P_{kj}^{(a)}$ has a move to $P_{kj}^{(b)}$, which is a 0, and to $P_{ij}^{(c)}$. $P_{kj}^{(c)}$ has a move to $P_{kj}^{(d)}$, which is again a 0, and to $P_{kj}^{(e)}$. $P_{ij}^{(e)}$ has a single move to $P_{kj}^{(f)}$, which in turn has a single move to $R)j$. As such, $P_{ij}^{(f)}$ will be either 0 or $\ast$. If it is 0, then $P_{kj}^{(e)}$ is $\ast$, making $P_{kj}^{(c)}$ $\ast 2$, making $P_{kj}^{(a)}$ be $\ast$. Otherwise, $P_{kj}^{(e)}$ will be 0, making $P_{kj}^{(c)}$ have value $\ast$, making $P_{kj}^{(a)}$ have value $\ast 2$.
%Again, as mentioned in the previous lemma, $R_j$ is 0 if and only if $N_j$ has been removed, and otherwise it is $\ast 2$ since moving to $N_j$ has value 0 by lemma \ref{lem:rToN}, and the other moves are to $\ast$. So, $P_{kj}^{(a)}$ is $\ast 2$ if $N_j$ has been removed, and is otherwise $\ast$ Thus, if any $N_j$ is removed, then $P_k$ is $\ast 3$. Otherwise, it is $\ast 2$.
\end{proof}

\begin{lemma}
If the token is on $N_k$, and only $N_i$ vertices with $i > k$ have been removed from the graph, then the nim-value must be at least $\ast 4$.
\label{lem:not1or2}
\end{lemma}
\begin{proof}
It is trivial for $k=4$, since by construction it has moves to 0 through $\ast 3$.

For larger $k$, they have a move to 0 through $R_j$ (by Lemma \ref{lem:baseZero}). They also have a move to $\ast 3$ by construction. The move to $M_k$ is $\ast$, and the move to $P_k$ is $\ast 2$ by Lemmas \ref{lem:mValue} and \ref{lem:pValue}, respectively.
\end{proof}

\begin{lemma}[Parity]
If the only vertices removed are $N_i$ vertices, the token is currently on $N_k$, $N_j$ is of lower rank than $N_k$ and is the lowest rank that has been removed, and if there are an odd number of $N_p$ vertices remaining, where $N_p$ are higher rank than $N_j$, then the nim-value of the game is $\ast$. If there are an even number of those vertices remaining, then the game has value $\ast 2$.
\label{lem:nGameValue}
\end{lemma}
\begin{proof}
We will do this by induction on the number of remaining $N_p$.

\textbf{Base Cases}: If the token is on an $N_p$ vertex, when it is the only one remaining, then the value is $\ast$, and when there are two $N_p$ vertices remaining, then the value is $\ast 2$.

To establish this: $N_k$ has a move to $\ast 3$ by construction, along with moves to $M_k$ and $P_k$ ($\ast 3$ by Lemmas \ref{lem:mValue} and \ref{lem:pValue}), a move to $R_k$, which is a move to 0 by Lemma \ref{lem:baseZero}, and to any $N_i$ with rank no more than $j$, which is at least $\ast 4$ by Lemma \ref{lem:not1or2}. Thus, there are no moves to $\ast$ and a move to 0, so the value is $\ast$.

In the case where there is one other $N_p$ vertex remaining, we have the exact same analysis, but there is a move to $\ast$ as well, by the previous base case.  Thus, the position has value $\ast 2$.

\textbf{Inductive Hypothesis}: If there are $x$ $N_p$ remaining, where $x$ is odd, then the value is $\ast$, and if $x$ is even, the value is $\ast 2$.

\textbf{Inductive Step}: If there are $x+1$ $N_p$ remaining,  where $x+1$ is odd, then the value is $\ast$, and if $x+1$ is even, the value is $\ast 2$.

To establish this: by the same argument as in the base cases, there is a move to 0, three moves to $\ast 3$, and a collection of moves with value at least $\ast 4$. The rest of the moves are to various $N_p$. By the IH, each of these options have value $\ast$ if $x$ is odd, and value $\ast 2$ if $x$ is even. Thus, if $x$ is odd, then the value with $x+1$ is $\ast 2$, and if $x$ is even, then the value with $x+1$ is $\ast$.
\end{proof}

\begin{theorem}[Right Amount of Stars]
When the token is on~$N_n$, the resulting game has nim-value $\ast n$.
\label{thm:constructability}
\end{theorem}
\begin{proof}
For 0 through $\ast 3$, it clearly works as we just build a tree as described by Observation \ref{obs:tree}.
For larger Grundy values, we have the token on vertex $N_n$. We will prove this has value $\ast n$ through induction on the values of a starting token on $N_i$.

\textbf{Base Case}: As long as the only vertices removed from the graph are $N_i$ vertices, the token on $N_4$ will have value $\ast 4$.

To establish this: %There are moves to $\ast$ ($M_n$), $\ast 2$ ($P_n$), and $\ast 3$.
there are moves to $\ast$, $\ast 2$, and $\ast 3$, each by construction.
There is a move to 0 through $R_4$ by Lemma \ref{lem:baseZero}. The only other available moves are some subset of of the $N_j$, which by Lemma \ref{lem:nGameValue}, have value $\ast$ or $\ast 2$.

\textbf{Inductive Hypothesis}: As long as the only vertices removed from the graph are various $N_i$ vertices where $i > k$, $N_k$ has value $\ast k$.

\textbf{Inductive Step}: We need to show that as long as the only vertices removed from the graph are various $N_i$ vertices where $i > k+1$, $N_{k+1}$ has value $\ast k+1$.

To establish this: $N_{k+1}$ has moves to $\ast 3$ and $\ast 4$, by construction, and to $M_{k+1}$, $P_{k+1}$, $R_{k+1}$, all of $N_4$ through $N_k$, and some of $N_{k+2}$ to $N_n$.  Moves to $N_4$ to $N_k$ are $\ast 4$ to $\ast k$ by induction. The move to $R_{k+1}$ is a move to 0 by Lemma \ref{lem:baseZero}. The move to $M_{k+1}$ is $\ast$, by Lemma \ref{lem:mValue} since all $N_j$ remain. The move to to $P_{k+1}$ is $\ast 2$, by Lemma \ref{lem:pValue}, again since no $N_j$ is removed.
\end{proof}

\subsection{Complexity Implication}
\label{Sec:HardClassifier}

We now use our polynomial-high nimber constructability result to prove Theorem \ref{Thm:Mys}, establishing other than the polynomial-time time “Zero”-“Fuzzy” classifier, every classifier of the Grundy values in \ruleset{Undirected Geography} is PSPACE-hard.

\begin{proof} (of Theorem \ref{Thm:Mys})
Recall that the proof above makes the distinction between $\ast$ and $\ast 2$ to be PSPACE-hard. We will first use this to prove that distinguishing between $\ast (k - 1)$ and $\ast k$ is PSPACE-hard, for any $k\geq2$.

We prove this by taking a position $(G_2,v_2)$ that is hard to distinguish between $\ast$ and $\ast 2$.
%which we will call $v_2$, and simply making
We introduce a new vertex $v_3$ with moves to its own $0$ and $\ast$ and add edge $(v_3,v_2)$ to create $G_3$.
%have a move to that position.
Then we will create a new vertex $v_4$ with moves to its own $0$, $\ast$, $\ast 2$, and connect  $(v_4,v_3)$ to create $G_4$, and so on, until we create a vertex $v_k$ with moves to its own 0 to $\ast k - 1$, and add edge $(v_k,v_{k-1})$ to create $G_k$.
These vertices $v_i$ and their associated gadgets have size polynomial in $i$ due to Theorem \ref{thm:constructability}.

Now, if $(G_2, v_2) = \ast$, then $(G_3, v_3)$ doesn't have a move to $\ast 2$, so $(G_3, v_3) = \ast 2$.  Similarly, $(G_4, v_4) = \ast 3$, $(G_5, v_5) = \ast 4, \ldots, (G_k, v_k) = \ast (k-1)$.  If instead, $(G_2, v_2) = \ast 2$, then $(G_3, v_3) = \ast 3$, because there is a move to $\ast 2$.  Likewise, $(G_4, v_4) = \ast 4, \ldots, (G_k, v_k) = \ast k$.  Thus, it is \cclass{PSPACE}-hard to distinguish between $\ast k$ and $\ast (k-1)$.

Next, we prove that distinguishing between any $\ast k$ and $\ast p$
is \cclass{PSPACE}-hard.  (We will assume $p > k$, without loss of generality.)  We first create a  $(G_k,v_k)$ where distinguishing $\ast k - 1$ and $\ast k$ is hard, then add a new vertex $v_p$ which has moves to its own 0 to $\ast k - 1$, $v_k$, and $\ast k + 1$ to $\ast p - 1$.  We name this graph $G_p'$; the position $(G_p', v_p)$ has value $\ast p$ exactly when $(G_k, v_k)$ has value $\ast k$.  $(G_p', v_p)$ has value $\ast k$ otherwise.  Thus, it is \cclass{PSPACE}-hard to distinguish between $\ast p$ and $\ast k$.

Finally, we use this to show that distinguishing between any possible fixed set of Grundy values is hard. For any possible set $S$, there must be at least one Grundy value $x \in S$ and one Grundy value in $y \in \bar{S} := [\Delta]\setminus S$. Then, we can, as described above, create a position where it's \cclass{PSPACE}-hard to distinguish between $\ast x$ and $\ast y$. Thus, if one could classify the game to be within that set of Grundy values, one could solve a \cclass{PSPACE}-hard problem.
\end{proof}

%\section{Final Remarks and Open Questions: Theory and Practice}

\section{Math Behind Board Games: Theory and Practice}

\label{Sec:FinalRemarks}

\begin{quote}
``{\em My experiences also strongly confirmed my
previous opinion that the best theory is inspired
by practice and the best practice is inspired by
theory.}'' - Donald E. Knuth \cite{KnuthTheoryPractice}
\end{quote}

Combinatorial game theory is a fascinating field, where simplicity is valued, and
both efficient methods for solving games and intriguing positions for challenging players are appreciated \cite{WinningWays:2001,LessonsInPlay:2007,SiegelCGT:2013,DBLP:books/daglib/0023750}.
Indeed, the magic smile on a six-year old's face when they realize a winning trick (e.g. how to win two-pile \ruleset{Nim}\footnote{when realizing the fact that \ruleset{Nim} with two identical piles is a losing position can be used for finding a winning strategy for any two-pile \ruleset{Nim}---including the decision to go first or second---so that they will never again lose to their parents.} as introduced in Math Circle\footnote{{\tt https://mathcircles.org/}}) is as enchanting as the contemplative gaze \cite{TheThinkersDavidLlada} of \ruleset{Chess}, \ruleset{Go}, and \ruleset{Hex} champions.
These are the polynomial-time smiles and \cclass{PSPACE}-hard gazes.

In this paper, we have proved that adding a single `PASS'---the smallest possible extension---to
\ruleset{Undirected Geography} transforms the game from
polynomial-time solvable to \cclass{PSPACE}-hard intractable. And similarly, we showed that giving a single pass to the game of \ruleset{Uncooperative Uno} also had this same transformation from P to PSPACE.
Characterizing the complexity impact of this small change to the ruleset has deepened and
expanded our understanding of the  foundational concept \& characterization in combinatorial game theory.
It has also added \ruleset{Multi-Token Undirected Geography} to the collection of \cclass{PSPACE}-hard graph-based impartial games
 with simple rulesets.

David Eppstein \cite{Eppstein} once eloquently expressed that elegant combinatorial games with simple, easy to understand \& remember rulesets yet intractable complexity are the gold standard for combinatorial game design.
His reason is a computational one: If a ruleset is polynomial-time solvable, then
optimal players can be programmed (or be replaced by an efficient computer program);
%. In a match with one of these players, there is no need to play the game out to determine whether you can win; just run the algorithm and see what it tells you.
thus intractable rulesets are essential
 to make competition interesting.
%In order to make the competition interesting, we want the winnability to be computationally intractable, meaning there’s no known efficient algorithm to always calculate a position’s outcome class. One way to argue that this is the case is to show that the problem of finding the outcome class is hard for a common complexity class. Many such combinatorial games are found to be PSPACE-hard, meaning that finding a polynomial-time algorithm automatically leads to a polynomial-time solution to all problems in PSPACE。
Our work suggests that an additional property that can be meaningful as part of the gold standard.  We call this the ``magic expression'' property: {\em there is a simple and natural perturbation to an intractable ruleset that makes it %from intractable to
tractable.}
The transformation from magic smile to pensive gaze can contribute to the computational-thinking dimension \cite{ComputationalThinking} in
the pedagogical value of recreational mathematics.
Fittingly, Sprague-Grundy theory---itself a general principle % for all disjunctive sum of impartial games
emerged from the polynomial-time solvable \ruleset{Nim} \cite{Bouton:1901}---has led us to the \cclass{PSPACE}-hard perturbation to \ruleset{Undirected Geography}, and a new complexity-theoretical understanding.

Naturally, the aesthetic quality of game boards is subject to individual taste.
To us, grid-like game boards---as those used in \ruleset{Hex}, \ruleset{GO}, \ruleset{Domineering}, \ruleset{Chomp}, and \ruleset{Atropos}---are attractive.
Most graphs are too complex visually for game boards;
we consider this to be one of the practical challenges, in implementing/popularizing combinatorial games on graphs.
Directed edges without well-recognized patterns can further add to the entropy.
On the other hand, simple graphs may reduce the strategic challenge of the game.
This is why the \cclass{PSPACE}-hardness of \ruleset{Two-Token Undirected Geography} and \ruleset{Undirected Geography with Passes} provides us with some excitement.
The removal of edge directions from \ruleset{Generalized Geography}, while retaining its \cclass{PSPACE}-hard complexity without introducing complex rules, opens up several possibilities.
Further, the simplicity of the graph condition in our Dichotomy Theorem (\ref{thm:dichotomy})---thanks also to the brilliant reduction of Lichtenstein and Sipser---brings us quite close to two-dimensional grids.
Affirmative answers to the following open questions will make these intractable extensions to %these are no longer open...
\ruleset{Undirected Geography} more elegant
for practical implementation.
\begin{openquestion}[Grid-Like \ruleset{Undirected Geography}]
Can the Grundy value in some version of (rectangle or hexagon) grid-based \ruleset{Undirected Geography} be \cclass{PSPACE}-hard to computed?
\end{openquestion}

We have been investigating whether the
``snap-to-grid'' result of Lichtenstein and Sipser for \ruleset{GO} can be extended to our case.
Meanwhile, inspired by these ``intractable'' \ruleset{Undirected Geography} extensions,
we have started to explore two practical designs, using the standard \ruleset{GO} or \ruleset{Hex} game boards.
We conclude the paper with a brief discussion about this two games and some theoretical questions they inspire.

\begin{itemize}
\item \ruleset{Binary Undirected Geography}: The game is based on \ruleset{Two-Token Undirected Geography}.
It can natuarlly be played with standard \ruleset{Go} and \ruleset{Hex} game sets.
``{\em On your turn, choose one of the empty (gray) nodes adjacent to either the last created black or white node. That node will become the last of its color.}''

We recently implemented this new game and the web-version
can be played with following link.
One can play either against another human player (sitting at the same computer) or some of our AI programs:
\begin{center}
\url{https://turing.plymouth.edu/~kgb1013/DB/combGames/twoBUG.html}.
\end{center}

We are cautiously optimistic that this simple game is challenging to play optimally. We are evaluating the optimal starting placements of the two tokens and with hope of finding the ``snap-to-grid'' hardness.

\item \ruleset{Navigate the Pass}: This game is based on \ruleset{Undirected Geography with Passes}.
Although both are \cclass{PSPACE}-hard games, the hardness of the single-pass version appears to be more brittle than
the two-token version in practice.
Aiming for a practical design with rectangle or hexagon grids,
we are in an early design stage, trying to characterizing the brittle patterns.

Consider the position with rectangle of side length $n$, e.g., the $19\times 19$ \ruleset{GO} board.
Imagine a white stone is located in grid point $(x,y)$. A player can either add a white stone to an adjacent grid point of the current stone, or call ``PASS'', after which, players can only play the black stone until all adjacent locations are occupied.
Note that, once switched, the game returned to standard \ruleset{Undirected Geography}, and hence the winnability can be determined in polynomial time by the matching test.
Furthermore, before switching the stone types, if $(x,y)$ is not in all maximum matchings, then the current player should play ``PASS''.

To make the game more robust, we can also play the game by placing stones in cells to expand the number of neighbors from four to eight, or to play in cells on a hexagon grid
(i.e. with six neighbors).
We are still investigating whether or not rectangle grids are tractable for this design and whether or not
the six-point or eight-point star stencil yields a more
 more challenging game than the four-point stencil on rectangular grids.
\end{itemize}

Mathematically, these questions  motivate us to look for other graph parameter(s) to capture the tractable-intractable divide.
For example, it is basic that Grundy values on trees in \ruleset{Undirected Geography} can be computed in
polynomial time.
It is well-known that the {\em treewidth}
of the $n\times n$ grid is $n$.

\begin{openquestion}[Fixed-Parameter Tractablility]
Are the Grundy values in \ruleset{Undirected Geography} fixed-parameter tractable in the treewidth of the graph?
\end{openquestion}

Note that playing extensions of \ruleset{Undirected Geography} on the grid points of hexagon grids involves
nodes with degree at most three.
So, \ruleset{Navigate A Pass} for this version can be solved in polynomial time.
It remains open whether or not \ruleset{Two-Token Undirected Geography} is tractable.

\begin{openquestion}[Two Bugs on Cubic Graphs]
Is there a polynomial-time algorithm to decide the winner of  \ruleset{Two-Token Undirected Geography} over degree-three graphs, or is the game \cclass{PSPACE}-complete to solve?
\end{openquestion}

We hope the theoretical questions about these simple grids
may lead us not only to elegant practical games but also to
something fundamental about computing and mathematical structures.

Earlier we mentioned the distinction between our \rsUndirGeog\ result and Morris', which provides hardness via a sum of many shallow partisan games.  We wonder whether there exists a partisan analog of our result.

\begin{openquestion}[Hard Partisan Sum]
Does there exist a well-known, strictly-partisan ruleset such that the winnability of a single instance can be solved in polynomial time, but where the winnability of a sum of two instances is \cclass{PSPACE}-hard?
\end{openquestion}

\bibliographystyle{plain}

\appendix

\section{Winning Nim by Nim-Sum}\label{AppendixNim}
Each \ruleset{Nim} position  consists of a collection of piles of items.
Two (or multiple players)
take turns picking items (at least one) from one of the piles.
Under normal play, the player who takes the last item wins the game.

As established by Bouton \cite{Bouton:1901}, \ruleset{Nim} is an exemplar sum game - each battlefield game is a single \ruleset{Nim}-pile.
\ruleset{Nim} is polynomial-time solvable (in the number of bits encoding positions) because of two %combined
properties: (1) the {\em nim-sum} is polynomial-time computable, and (2) the nimber of a single \ruleset{Nim}-pile is embarrassingly easy to calculate.
Thus, the next player has a winning strategy in  \ruleset{Nim} if and only if the bitwise-exclusive-or of the binary representation of the pile sizes is not zero.

This polynomial-time solution to \ruleset{Nim} inspired Sprague-Grundy theory, which applies to all impartial games.
%serves as a beautiful illustration of its potential effectiveness both  in practice.
In optimization, related problems can often  be understood through adding  or removing constraints or modifying an objective function, providing a systematic way of deriving algorithms and intractability for using the original results \cite{PapadimitriouBook:1994}.
With impartial games, the Sprague-Grundy
theory is a tool with similar utility.

\section{Winning Undirected Geography by Matching}
\label{Appendix:Matching}
%Proof of Theorem \ref{Theo:LSU}}\label{App:Proofs}
For completeness and our analysis %in the subsequent sections,
 we include a proof for this classical result.

\begin{theorem}[A Matching-Based Characterization]\label{Theo:LSU}
For any undirected graph $G=(V,E)$  and $s\in V$ satisfying
  $E\neq \emptyset$,
 \ruleset{Undirected Geography} at  $(G,s)$ is a winnable position if and only if $s$ is in every maximum matching of $G$.
\end{theorem}

\begin{proof}
For a subset $S\subset V$, let $G_{S}$ be the graph obtained from $G$ by removing $S$ and all edges incident to $S$.
First note that $s$ is not in every maximum matching of $G$
  if and only if the maximum-matching size of $G$ is  equal to
  that of $G_{\{s\}}$.
Thus, we can efficiently determine this graph property  using
standard polynomial-time maximum matching algorithm.

We now prove the following: Suppose $G$ is not empty. Then:

\begin{enumerate}
\item If $s$ is in every maximum matching of $G$, then $s$ has a neighbor $v$ that is not in some maximum matching of $G_{\{s\}}$.
\item If $s$ is not in every maximum matching $G$, then either $s$ has no neighbor in $G$, or every maximum matching in $G_{\{s\}}$ contains all neighbors of $s$.
\end{enumerate}

To see 1, consider any maximum matching $M$ of $G$ (say the one computed by a polynomial-time matching algorithm).
Because $G\neq \emptyset$, $M$ is not empty.
We now prove that $M \setminus (s,M(s))$---which clearly does not contain $M(s)$---is an maximum matching
of $G_{\{s\}}$.
To see the maximality, let $M'$ be a maximum matching of $G_s$.
If $|M'| = |M|$, then $M'$ is a maximum matching of $G$ without containing $s$, contradicting the earlier assumption.
Thus, $|M'| = |M|-1$, and $M \setminus (s,M(s))$ is
 a maximum matching in $G_{\{s\}}$,

 To see 2, first, it follows from the assumption that
$s$ is not in every maximum matching of $G$, the size of maximum matching of $G$ is equal to the size of the maximum matching of $G_s$.
Now suppose there exists a maximum matching $M'$ in $G_{\{s\}}$ that does not contain a neighbor, call it $v$, of $s$. Then $M'\cup (s,v)$ is also a matching of $G$, which contract to the statement above.

Now, we we show that if $s$ is in every maximum matching of $S$,
then $(G,s)$ is a winning position.
Let $M(s)$ be the node for which $(s,M(s))$ is in matching $M$.
We now prove that selecting $M(s)$ is a winning move.
% that is,  $(G_{\{s\}},M(s))$ is a losing position.
We showed earlier that $M \setminus (s,M(s))$ is a maximum matching of $G_{\{s\}}$. Furthermore, it
 contains every neighbor $u$ of $M(s)$ in $G_{\{s\}}$.
By the same argument, $u$ is in every maximum matching of
$G_{\{s,M(s)\}}$.
The theorem then follows from a proof by induction using the above analysis as the induction step.
\end{proof}

\section{Finding Grundy Values By Branch-And-Bound}

In this section, we extend the branch-and-bound Grundy-value evaluation process to the abstract setting as well as
analyze the impact of high-degree nodes in the process.

But first, as an illustration of what has been been established in Section \ref{sec:BranchAndBound}, consider the following family of fun planar \ruleset{Undirected Geography} games:

\begin{corollary}[Alternation Paths Through Fully-Triangulated Planar Maps]
For any fully triangulated planar graph $G$, and a face $f$ in $G$, the Grundy value of the \ruleset{Undirected Geography} game on the {\em dual} of $G$, starting at $f$, can be computed in polynomial time.
\end{corollary}

\subsection{Nimber-Winnability Reduction: Degree of Phase Transition to Infeasibility}

\label{sec:Abstract}

Our dichotomy characterization of Grundy-value
computation can be  extended from the concrete \ruleset{Undirected Geography} to
  an abstract nimber-winnability reduction in general impartial games, characterized by their {\em degrees} and {\em heights}.

\begin{itemize}
\item {\bf Degree:} For a positive integer $\Delta$, we say that an impartial game $g$ is a degree-$\Delta$ game if $g$ and all positions reachable by $g$ have at most $\Delta$ feasible moves.

\item {\bf Height:} For a positive integer $h$, we say an impartial game $g$ is a height-$h$ game if the height of $g$'s' game tree is at most $h$.
We say $g$ is a {\em polynomially-short} game if the height of its game tree is upper bounded by a polynomial function of its input size.
\end{itemize}

For example, consider an undirected graph $G=(V,E)$  with $n$ nodes and maximum degree $\Delta$, and a node $s\in V$ with degree less than $\Delta$.
Then, \ruleset{Undirected Geography} at position $(G,s)$ is a game with degree-$(\Delta-1)$ and height-$n$.

\begin{theorem}[Dichotomy of Nimber-Winnability Reduction]
For any degree-two, polynomially-short impartial games, Grundy-value computation can be reduced in polynomial-time to decision of winnability.
In contrast, there exists a family of degree-three, polynomially-short impartial games for which the winnability can be solved in polynomial time, but Grundy-value computation is \cclass{PSPACE}-hard intractable.
\end{theorem}
\begin{proof}
In the branch-and-bound process at each node of the game tree encountered, we first run a winnability test for the position.
If it is ``Zero'', then return $0$.
Otherwise, run a winnability test for each of their children in the game tree, and we know that one of the must be ``Zero''.
If both are ``Zero'', then return $*$.
Otherwise, we following the winning way to determine whether the other children is $*$ or $*2$, and return $*2$ or $*$ accordingly.
\end{proof}

\subsection{The Impact of Large-Degree Nodes in Branch-and-Bound}\label{Appendix:Extensions}

We can extend the analysis of Section \ref{sec:BranchAndBound} to show that
a few nodes with large degrees will not stop the polynomial-time branch-and-bound:

\begin{theorem}[A Tractable Nimber Terrain in \ruleset{Undirected Geography}]
For any constants $D$ and $\Delta$, $c$,
if $G =(V,E)$ is an undirected graph with $n=|V|$ nodes, in which
at most $c\log_2 n$ nodes with degree in range $[4,\Delta+1]$,
and at most $D$ nodes with degree more than $\Delta$,
then the Grundy value of the \ruleset{Undirected Geography} game over $G$
can be computed in time $O(n^{D+ c\log \Delta+3})$.
\end{theorem}
\begin{proof}
We will apply the ``following the winning way'' technique at all degree-two game-tree nodes in the standard DFS-based recursive evaluation methods.
(1) Each time when we evaluate a node with degree more than $\Delta$, the
branching factor is at most $n$. (2) Each time when we evaluate a node with degree in range $[4,\Delta+1]$, the branching factor is at most $\Delta$ (or $\Delta + 1$ for such  starting node). (3) Otherwise, at degree three node, the branching factor is one.
If we use an $O(n^3)$ time algorithm for maximum matching, then we can bound the total time by:
$O\left(n^D\cdot \Delta^{c\log n}\cdot n^3\right) = O\left(n^{D+ c\log \Delta+3}\right).$
\end{proof}

%%%%%%%%%%%%%%%%%%%% Alternate Proofs %%%%%%%%%%%%%%%%%%%%%%%%
\section{Alternate Gadget Proofs: Winnability when added to $\ast$}
\label{sec:alternateProofs}

In this section, we present our alternative proofs from our reduction, using figure \ref{fig:directedEdgeGadget}.  These proofs are characterized by preserving winnability when we add the reduced game to $\ast$, a single move that can be used once by either of the players.

We will refer to two players as the Foe, who makes losing moves, and the Hero, who will respond with a winning strategy.  We will exhaustively describe the Hero's strategy to force a win.

\begin{lemma}[Wrong Way]
    Moving from $y$ to any vertex $d$ results in a value of $\ast 2$ or $\ast 3$.  In other words, $(G'_y, d) = \ast 2$ or $\ast 3$.
    \label{lem:wrongWayAlt}
\end{lemma}

(This lemma was originally stated in lemma \ref{lem:wrongWay}.)

\begin{proof}
Moving from $y$ to $d$ on our main component needs to result in a losing position when added to $\ast$.  By showing that it's a losing move, we will show that there is a winning response, meaning the sum is non-zero.  That will mean that the main component cannot be equal to $\ast$ (because $\ast + \ast = 0$).  That component also cannot be equal to 0 because there is a terminal move to $d_0$.  Since there are two other possible options ($f$ and $c$), it can either be $\ast 2$ or $\ast 3$.

On this sum, if a player, the Foe, chooses to move $y$ to $d$, we need to show that this is a losing play.  The winning response, for the other player, the Hero, is to move to $f$.  If the Foe moves to $b$, then the Hero plays on the $\ast$.  Now, no matter whether the Foe moves to $a$ or $c$, the Hero can respond by moving to $a_0$ or $c_0$ to win.

If the Foe takes the $\ast$ at $f$ instead of moving to $b$, then the Hero responds by moving to $b$, putting them in the same situation as above.
\end{proof}

\begin{lemma}[Right Way]
    Moving from $d$ to $y$ results in a value of $\ast$ exactly when moving from $x$ to $a$ results in $\ast$.
    \label{lem:rightWayAlt}
\end{lemma}

(This lemma was originally stated in lemma \ref{lem:rightWay}.)

\begin{proof}
    The statement is equivalent to saying that after adding to $\ast$, moving $d$ to $y$ results in an overall value of 0 exactly when the same is true moving $x$ to $a$.
    We prove this winnability by giving an odd-length sequence of moves from $x$ to $y$ and show that deviating from this plan is always a losing move.  We refer to the deviating player as the Foe, and their opponent as the Hero.

    That sequence is: $x \rightarrow a \rightarrow b \rightarrow c \rightarrow d \rightarrow y$.  Since there are odd moves in this, the player who moves to $a$ is the same player that moves to $y$, meaning this mimics the behavior of a directed \rsGeog\ edge.  We complete the proof by exhaustively showing that all deviations are losing moves.

    Note that any move to a terminal vertex ($a_0, c_0, d_0$) is a losing move because the Hero can play on the added $\ast$ to make the sum zero.  The other deviations are moving to $f$ from either $b$ or $d$.  Moving from $d$ is also a terminal position (because $b$ has already been taken) so that is a losing move.  Finally, if the Foe moves $b$ to $f$, the Hero responds by playing on the $\ast$.  Then the Foe must move to $d$ and the Hero can move to $d_0$ to win.
\end{proof}

\begin{theorem}[Complexity Separation of Winnability and Grundy Values]
\label{thm:nimberHardnessAlt}
%[Intractability]
The Grundy value of polynomial-time solvable \ruleset{Undirected Geography}
   is \cclass{PSPACE}-hard to compute
  in planar bipartite graphs of maximum degree four.
\end{theorem}

(This theorem is originally stated as theorem \ref{thm:nimberHardness}.)

\begin{proof}
    This version of the proof uses the add-to-$\ast$ characterization.  For this, we remark that the $\outP$ \rsGeog\ positions are exactly those where the sum is equal to zero.  Since it's always a losing move to go backwards to a $d$ vertex, by lemma \ref{lem:wrongWay}, and moving the right way to an $a$ vertex is a winning move exactly when moving to that $y$ vertex is a winning move, by \ref{lem:rightWay}.  Thus players can only win by following the direction of the gadgets and not by playing on the $\ast$ component.

    At the end, when there are no directed edge gadgets to from current vertex $v$ that lead to a vertex $y$, the value of the sum is 0 because the next player can either move to an $a$ without a $y$, to a $d$, or to $v_0$.  The value of just the \rsUndirGeog\ component is 0, so the opponent can play on the $\ast$-component and end the game.

    The winning moves in $(G, s)$ correspond exactly to winning moves in $(G', s) + \ast$, so the reduction works.
\end{proof}

\section{The Bipartite Graphs of Uno}\label{Appendix:Uno}

In 2014, \cite{demaine2014uno}, among proving other things, included a reduction from a game of 2 player Uno called \ruleset{Uncooperative Uno} to \ruleset{Undirected Geography}. Notably for this reduction, we can prove that this reduced version is still sufficient to demonstrate hardness for the Grundy value.

First, let's formally define this ruleset. In this game, there are two hands, $H_1$ and $H_2$, which each consist of a set of cards. This is a perfect information game, so both players may see each other's hands. Each card has two attributes, a color $c$ and a rank $r$. Each card then thus be represented $(c, r)$. A card can only be played in the center (shared) pile if the previous card matches either the $c$ of the current card or the $r$ of the current card.

Demaine  {\em et al} \cite{demaine2014uno} then gave a simple reduction from this to \ruleset{Undirected Geography}, which, since it isomorphically preserves options, also preserves the Grundy value. The reduction is to simply represent the game as a bipartite graph, where each card is represented by a vertex, and the cards in $H_1$ are in one part of the partition while the cards in $H_2$ are in the other. There is an edge between the vertices if and only if the card could be played in response to the other. Thus, the game is simply an \ruleset{Undirected Geography} game played on the bipartite graph.

Of course, while this is a nimber-preserving reduction, it isn't clear that the reduction works in the other direction, which is what we need.  This is because, in addition to being bipartite, the graph for hard instances of \ruleset{Undirected Geography} in Grundy-value computation would additionally need to be constrained to only have edges based on the matching color/rank principle of \ruleset{Uncooperative Uno}.

More formally, for a target bipartite graph to be reduced to \ruleset{Uncooperative Uno}, we need the graph to have the following properties:

\begin{enumerate}
\item Each vertex $v$ can be given a pair of integers $(a, b)$. We will refer to the $a$ of $v$ as $a(v)$ and similarly $b(v)$ just means the $b$ of v.
\item A vertex $v_1$ in partition $P_1$ is adjacent to a vertex $v_2$ in $P_2$ if and only if $a(v_1) = a(v_2)$ or $b(v_1) = b(v_2)$
\end{enumerate}

Not all bipartite graphs have these properties, including the hard instance, if the directed edges were made undirected, from Lichhtenstein's and Sipser's reduction. Fortunately, our edge gadget from Theorem \ref{thm:nimberHardness} (with some very slight modifications that preserves the proof and nimber) transforms any directed bipartite graph into one has those properties. Therefore, to show hardness for the nimber of \ruleset{Uncooperative Uno}, we may start with Lichhtenstein's and Sipser's bipartite graph, apply our reduction, with a slight modification, and then reduce from this instance of \ruleset{Undirected Geography} to \ruleset{Uncooperative Uno}

We will call this modified reduction, which we will discuss and motivate in the proof, as $H = (V_H, E_H)$ while the original version as it appears in the theorem will be $G = (V_G, E_G)$.

\begin{lemma}
In $H$, one can assign a pair of integers $(a, b)$ to each vertex $v$ such that for all $v$, for each vertex $u$ that is adjacent to $v$, either $a(v) = a(u)$ or $b(v) = b(u)$, and there exists no vertex $w$ adjacent to $v$ that is in the opposite bipartite partition to $v$ that has $a(v) = a(w)$ or $b(v) = b(w)$
\end{lemma}
\begin{proof}
We will work on this proof as if it is for $G$, and then switch to $H$, which is just a small modification, at the end of this proof.

One strategy is, when labeling vertices in our reduction, to simply alternate between having $a(v) = a(u)$ and $b(v) = b(u)$. For the first vertex, $v$, we assign $a$ to an integer (and $b$ to be some arbitrary number that will never appear for a vertex in the opposite partition). Then we can assign all of the neighbors $u$ have the property $a(v) = a(u)$. Then all of their neighbors, $w$, each have a unique $b$ such that $b(u) = b(w)$. Then, the neighbors of each vertex $w$. $x$ each have unique a $a$ such that $a(w) = a(x)$, and so on, in a breadth first search fashion.

Note that this works whenever all of the adjacent vertices have a completely different set of adjacent vertices. Similarly, there are no issues if the adjacent vertices share all of the same neighbors, since we can just have each of the vertices in that ``round'' have the same $a$ or $b$.

There are only two possible ways for this to break. The first occurs when $v$ and $u$ are both attempting to connect to $w$ or a vertex already covered by the other, but the round on the path (that is to say, whether $a$ or $b$ was used) doesn't match up. This can't happen, because our graph is bipartite and only one of $a$ or $b$ will be chosen for any path to that vertex.

The other possibility is that two vertices have a non empty intersection of vertices that they are both adjacent to, but they don't share all the vertices they are adjacent to. This happens exactly for two (repeated) segments in the graph from our reduction. The first is at the diamonds $b, c, d, f$, where a path moving from $d$ or $b$ to $f$ and $c$ will result in this situation. The second is for each $x$ and $y$ vertices, when coming from the adjacent $d$ and $a$ vertices, since they each have their own $a_0$ or $d_0$, while sharing the $x$ or $y$.

The second can be resolved by simply having the $d_0$ or $a_0$ use $a$ to connect if $x$ or $y$ use $b$, or similarly using $b$ if they use $a$. This works since the previous vertex in the sequence will be in the same partition as $d_0$ or $a_0$ and thus not break the rule.

The first case is why we must use $H$ instead of $G$. We will add a vertex between $b$ and $f$ and one between $f$ and $d$. This changes nothing about the reduction (each statement from the proof of Theorem \ref{thm:nimberHardness} and this one still holds). But now, we can apply the same solution as we did previously, and simply choose the same component for $c_0$ as was used to connect to $c$ previously.
\end{proof}

The reduction to \ruleset{Uncooperative Uno} is now easy to complete, as the pair of labels can now be used for the color and rank of cards, and the partitions as the players' hands. These games are now exactly isomorphic

\end{document}